\documentclass[12pt]{article}
\usepackage[utf8]{inputenc}

\usepackage{amsmath}
\usepackage{amstext}
\usepackage{amssymb}
\usepackage{amsfonts}
\usepackage{geometry}
\usepackage[pdftex]{graphicx}
\usepackage{indentfirst}
\usepackage{appendix}
\usepackage{float}
\usepackage{setspace}
\usepackage{hyperref,cleveref}
\usepackage{aliascnt}
\usepackage{bm}
\usepackage{comment}
\usepackage[authoryear,round,sort]{natbib}
\usepackage{mathtools}

\usepackage[moderate]{savetrees}
\usepackage{times}
\usepackage{xcolor}

\hypersetup{colorlinks=true,
            citecolor=blue,
            linkcolor=blue,
            urlcolor=blue,
			bookmarksopen=true,
			pdfstartview={XYZ null null 1.00},
			pdfpagelayout={SinglePage}
}

\newtheorem{theorem}{Theorem}
\newtheorem{assumption}{Assumption}

\newtheorem{corollary}{Corollary}
\newtheorem{definition}{Definition}
\newtheorem{lemma}{Lemma}
\newtheorem{proposition}{Proposition}
\newtheorem{remark}{Remark}
\newenvironment{proof}{{\noindent\textbf{Proof}.\quad}}{\hfill $\blacksquare$\par}

\title{
\textbf{Dividing the Spoils:\\ Strategic Prize Allocation in Team Contests}\thanks{We thank Bo Chen, Yi-Chun Chen, Shanglyu Deng, Christian Ewerhart, Igor Letina, Yingkai Li, Zihao Li, Wojciech Olszewski, Harry Pei, Marco Serena, Satoru Takahashi, Allen Vong, Zhewei Wang, Zijia Wang, Julian Wright, Zenan Wu, Junjie Zhou, Yuxuan Zhu, and the participants of the NUS Theory Group Lunch Workshop, the 2023 SAET (Paris), the 2024 Game Theory World Congress (Beijing), the 2024 EAGT (Jeju) and the 2025 SDU IO Summer Workshop (Jinan), for their helpful comments and suggestions. Junhao Chen and Xinya Li provided excellent research assistance. AI-assisted software was used to improve language and readability. The authors reviewed and edited all output and take full responsibility for the content of the manuscript. Any remaining errors are our own.}}
\author{
Zhonghong Kuang\thanks{Zhonghong Kuang, School of Economics, Renmin University of China, 59 Zhongguancun Street, Beijing, China, 100872. \emph{E-mail:} kuang@ruc.edu.cn.}
\and
Jingfeng Lu\thanks{Jingfeng Lu, Department of Economics, National University of Singapore, Singapore, 117570. \emph{E-mail:} ecsljf@nus.edu.sg.}
\and
Yiyao Zhu\thanks{Yiyao Zhu, Department of Economics, National University of Singapore, Singapore, 117570. \emph{E-mail:} zhuyiyao@u.nus.edu.}}

\date{July 2026}

\begin{document}
\maketitle
\begin{abstract}
\noindent Rival teams compete through battles and reward members only after team victory. We study how team managers allocate victory-contingent reward budgets across players to maximize their teams' winning probabilities in a majoritarian contest. Under a regularity condition on battle technologies, the allocation game has a unique pure-strategy equilibrium. Despite differences in budgets and player costs across teams and technologies across battles, both managers choose the same normalized reward schedule, a property we call reward schedule alignment. Each battle's share is proportional to discriminatory power, closeness, and pivotality. With precommitted rewards, allocations and team-winning probabilities are independent of battle order.

\medskip



\end{abstract}
\newpage{}

\section{Introduction}
\label{Section:Intro}

Many competitions are won collectively but fought locally. A procurement consortium must perform across technical dimensions to secure a contract, while a national squad accumulates pairwise victories to win a title. In both settings, opposing members meet in separate battles, where costly effort shapes their winning probabilities, and local victories determine the collective outcome. Before effort is chosen, a team principal, such as a prime contractor or sports federation, may precommit how a victory-contingent reward budget is divided among members. We call this principal the manager. Because the rival manager makes her own commitment, dividing the collective prize is itself strategic.

To see the strategic tension, consider a three-battle contest in which team~$A$ is slightly stronger in two battles and much weaker in the third. A natural response is for $A$ to concentrate its rewards on the two winnable battles, while team~$B$ treats the third as secure and targets one of the other two. More generally, asymmetric teams may seem likely to specialize. But reward schedules are chosen strategically, and each manager anticipates the other's choice. The key question is whether this interaction reinforces specialization or overturns it.

We study this problem in a majoritarian team contest \citep*{Fu-2015-Team}. Two rival teams compete through an odd number of pairwise battles, with a majority determining team victory. Each manager seeks to maximize her team's probability of victory. Before effort is chosen, she publicly allocates an earmarked reward budget among her players, payable only if her team wins. Players then choose costly effort in their assigned battles. Battle technologies are homogeneous of degree zero, so only relative effort matters; the Tullock power form is the leading example. Under a regularity condition on battle technologies, the allocation game has a unique pure-strategy Nash equilibrium. Yet the equilibrium contradicts the specialization intuition: after normalization by total rewards, the managers choose exactly the same reward schedule.

We call this result \emph{reward schedule alignment}. Differences in budgets, player costs, and battle technologies, such as discriminatory power and bias, determine which battles receive larger shares, but both managers rank the battles in the same way. Strategic competition therefore produces mirroring rather than specialization.

Alignment reduces the allocation problem to a common assessment of which battles are most worth incentivizing. That assessment reflects three distinct forces, none sufficient on its own. Their interaction is summarized by a single index, which we call a battle's \emph{salience}. In equilibrium, each battle's share of either manager's reward budget is proportional to its salience:
\[
   \text{salience}
   = \text{discriminatory power}
     \times \text{closeness}
     \times \text{pivotality}.
\]
The three components have distinct economic roles. \emph{Discriminatory power} captures how strongly the battle technology translates relative incentives into changes in winning probabilities. \emph{Closeness}, defined as the product of the two teams' equilibrium battle-winning probabilities, measures how much the outcome remains in play and is greatest when the battle is evenly matched. \emph{Pivotality} is the probability that the remaining battles split evenly, in which case the focal battle determines the majority. Their product therefore combines incentive responsiveness, competitive balance, and importance for team victory. A battle receives a large reward share only when additional incentives can materially shift its outcome and that outcome is likely to decide the contest. A close but rarely pivotal battle, or a pivotal battle that responds little to incentives, receives a small share.

The rule can be computed directly. Reward schedule alignment fixes the reward ratio in every battle at the ratio of the two total budgets. This pins down battle-winning probabilities and pivotalities from budgets, player costs, and battle technologies. Each manager then allocates her budget across battles in proportion to salience.

Why does strategic competition generate alignment? An additional unit of reward matters through two channels. It changes the probability of winning the focal battle, and that change affects the probability of team victory only when the remaining battles split evenly. This pivotal event is the same for both teams. Moreover, homogeneity of degree zero means that only relative effort affects the local winning probability; in equilibrium, its response to rewards is governed by relative incentives. Comparing the managers' optimality conditions therefore eliminates team-specific scale factors and leaves them with the same marginal ranking of battles. In equilibrium, the reward ratio between the teams is consequently constant across battles and equals the ratio of their total budgets. Differences in budgets scale the two schedules, while player costs and battle technologies change the salience of a battle for both managers rather than tilting their schedules in opposing directions.

The main technical challenge is global existence. The first-order conditions identify a candidate allocation but do not by themselves establish a global best response, because a manager's majority-winning probability can have an irregular shape. We show that, under the condition imposed on battle technologies, each manager's team-winning probability is log-concave in her own reward allocation. Log-concavity makes every local optimum global and validates the salience characterization. This argument yields existence and uniqueness for any odd number of battles. For the Tullock power technology, the condition reduces to discriminatory power no greater than one.

The model yields two further implications. First, strengthening a team, either by lowering one member's marginal cost or by increasing the team's reward budget, raises its equilibrium probability of victory even after rewards are reallocated across battles. Second, when reward schedules are committed in advance, neither the equilibrium allocation nor either team's chance of winning depends on battle order. Simultaneous battles, fully sequential battles, and battles organized in sequential clusters therefore generate the same equilibrium reward allocations and ex ante team-winning probabilities. This timing invariance need not survive if managers can revise later rewards after observing earlier results.

These results come with clear boundaries. The model allows unequal budgets and substantial heterogeneity in player costs and battle technologies, but team victory is determined by majority rule, each player is assigned to one battle, reward budgets are exogenous, and rewards are committed in advance and paid only after team victory. Procurement workshares and sports bonus pools provide the clearest applications because rewards can be specified before effort and made contingent on collective success. Legislative elections permit a looser relational-contract interpretation in which committee positions, agenda access, or future party support serve as rewards contingent on majority control. We discuss these applications and their limitations in \autoref{subsec:implication}.

Our contribution is closest to the literature on team contests with pairwise battles. \citet*{Fu-2015-Team} take player stakes as given and characterize equilibrium effort under majority rule and arbitrary temporal structure. We add a prior strategic stage in which rival managers choose those stakes. This additional stage makes the cross-battle distribution of incentives endogenous and yields reward schedule alignment. \citet*{Feng-2024-Optimal} study how a centralized organizer divides a total prize across individual battles, whereas the schedules here emerge from competition between opposing principals. Related work examines tug-of-war contests \citep*{Hafner-2017-Tug,Hafner-2020-Eternal}, winner-effort maximization \citep*{Barbieri-2024-Winners}, and equilibrium player ordering \citep*{Konishi-2022-Equilibrium}.

The paper also relates to multi-battle contests \citep*{Harris-1987-Racing,Klumpp-2006-Primaries,Konrad-2009-Multi,Gelder-2014-From,Feng-2018-How,Clark-2020-Creating} and Colonel Blotto contests \citep*{Roberson-2006-Colonel,Kovenock-2021-Generalizations}. These models typically let contestants allocate productive resources directly across battlefields. Our managers instead allocate victory-contingent rewards, and strategic players respond by supplying costly effort. This distinction links resource allocation across battles to incentive design within teams and produces reward schedule alignment rather than specialization across battlefields.

Finally, the paper relates to prize design in contests \citep*{lazear1981rank,moldovanu2001optimal,olszewski2016large,fang2020turning,lemus2025contingent}, prize sharing in group contests \citep*{Nitzan-2011-Prize,Nitzan-2014-Intra,Kobayashi-2024-Effort,Konishi-2024-Allocation}, and moral hazard in teams \citep*{holmstrom1982moral,rayo2007relational,winter2004incentives,halac2021rank}. This literature generally studies how a single principal or organizer designs rewards to manage incentives. We study two principals who simultaneously design their internal reward schedules while anticipating each other's choices. Majority aggregation adds a further force absent from standard team incentive problems: the value of motivating a player depends on the probability that her battle determines collective victory. The resulting pivotality channel also connects the paper to electoral targeting \citep*{penrose1946elementary,shapley1954method,banzhaf1965weighted,brams1974three,dixit1996determinants,stromberg2008electoral,kujala2020donors}.

The remainder of the paper is organized as follows. \autoref{Section:Model} sets up the model. \autoref{Section:Equi} characterizes the unique pure-strategy equilibrium. \autoref{sec:discussion} discusses the role of the technological assumptions, develops comparative statics, and presents applications. \autoref{Section:Conclusion} concludes. The Appendix contains technical proofs.

\section{Model}\label{Section:Model}

Two rival teams, $A$ and $B$, play a two-stage game: each team's manager first divides a reward budget across her players, and the players then exert effort and fight the battles. The teams contest $2N+1$ pairwise battles ($N\geq 1$), indexed by $t\in\mathcal{N}\triangleq\{1,\ldots,2N+1\}$. Each team fields one risk-neutral player per battle and is run by a single manager. The pairing is fixed and commonly known: in battle~$t$, player~$i(t)$ of team~$i\in\{A,B\}$ exerts effort at constant marginal cost $c_{it}>0$. Conditional on effort choices, battle outcomes are independent across battles. A team wins the contest by winning at least $N+1$ battles. The baseline model takes the battles to be simultaneous; \autoref{lem:temporal} below shows that, for precommitted allocations, the order of play does not affect ex ante team-winning probabilities.

Manager~$i$ controls an earmarked reward fund $W_i>0$. The fund is contractually restricted to rewards for members of team~$i$, becomes payable only if team~$i$ wins, and cannot be retained by the manager or diverted to another organizational use. Before effort is chosen, the managers publicly and simultaneously commit to reward allocations
\[
  \mathbf{v}_i=(v_{i1},\ldots,v_{i,2N+1})\in\Delta_i
  \;\triangleq\;
  \Bigl\{\mathbf{v}\ge 0:
  \textstyle\sum_{t\in\mathcal{N}}v_t=W_i\Bigr\},
\]
where $v_{it}\ge 0$ is paid to player~$i(t)$ if team~$i$ wins. Each manager chooses only the distribution of her predetermined fund and seeks to maximize her team's probability of victory. Because the fund has no residual value to the manager, every feasible allocation exhausts $W_i$. The fund may represent, for example, a championship bonus pool fixed by a federation or a pool of workshare claims that a procurement consortium has committed to distribute among its participating units if it secures the contract. The budgets need not be equal; $W_A=W_B$ is the equal-budget benchmark.

Having observed $(\mathbf{v}_A,\mathbf{v}_B)$, the two players in each battle~$t$ simultaneously choose efforts $e_{At},e_{Bt}\ge 0$. Player~$A(t)$ wins the local battle with probability $\widetilde{p}_{At}(e_{At},e_{Bt})$, and player~$i(t)$ receives payoff $v_{it}\mathbf{1}\{\text{team $i$ wins}\}-c_{it}e_{it}$. The primitive battle success function may be biased, so equal efforts need not imply equal winning probabilities.

\paragraph{The contract space and commitment.}
Three modeling choices deserve comment. First, rewards are contingent on \emph{team} victory because they are claims on the spoils of collective success: a procurement workshare, a championship bonus pool, or the organizational rents associated with majority control become available only when the team wins. Within the present contract space, allowing an optional payment in the losing state would not improve incentives. Such a payment reduces the player's payoff difference between team victory and defeat, and a manager who may leave the corresponding resources uncommitted therefore sets it to zero.\footnote{Suppose player~$i(t)$ receives $v_{it}\ge 0$ if team~$i$ wins and an optional payment $u_{it}\ge 0$ if it loses. His payoff can be written as $u_{it}+(v_{it}-u_{it})\mathbf{1}\{\text{team $i$ wins}\}-c_{it}e_{it}$. Holding $v_{it}$ fixed, a larger $u_{it}$ lowers the incentive differential $v_{it}-u_{it}$ and therefore weakens effort incentives. If the manager can leave these resources uncommitted and values only the probability of team victory, she optimally chooses $u_{it}=0$. This argument does not cover compulsory losing-state budgets, participation constraints, limited liability, risk sharing, or other contractual roles for payments following defeat.}

Second, conditional on team victory, we restrict attention to reward amounts fixed before effort rather than schedules contingent on the realized profile of battle outcomes. This benchmark fits settings in which workshares or bonus amounts are announced ex ante and allows us to isolate the strategic interaction \emph{between} the rival managers. Outcome-contingent sharing within the winning team would introduce an additional tournament-design problem and is left for future work (\autoref{Section:Conclusion}).

Third, in the simultaneous baseline, commitment simply requires managers to choose rewards before effort, when no interim battle outcomes are available. Under sequential timing, commitment becomes a substantive restriction: publicly announced bonus pools or signed teaming agreements cannot be revised after early outcomes are observed. \autoref{lem:temporal} establishes timing invariance under such precommitment; history-contingent renegotiation is discussed in \autoref{Section:Conclusion}.

\subsection{Contest Technology}

Throughout, $\widetilde{p}_{At}+\widetilde{p}_{Bt}=1$, and we maintain the following regularity conditions on $\widetilde{p}_{At},\widetilde{p}_{Bt}$.
\begin{assumption}\label{ass:battle_technology}
For every battle $t\in\mathcal N$:
\begin{enumerate}
    \item For strictly positive efforts, the primitive success function $\widetilde p_{At}$ is \textbf{homogeneous of degree zero}: it can be represented by a twice continuously differentiable log-odds function $\ell_t:\mathbb R\to\mathbb R$ with $\ell_t'>0$,
    \[
      \ell_t(z_t)\triangleq\log\frac{\widetilde p_{At}}{\widetilde p_{Bt}},
      \qquad
      z_t\triangleq\log(e_{At}/e_{Bt}),
      \qquad
      \widetilde p_{At}=\frac{e^{\ell_t(z_t)}}{1+e^{\ell_t(z_t)}}.
    \]
    \item \textbf{Boundary regularity} holds: $\lim\limits_{z\to-\infty}\ell_t(z)=-\infty$ and $\lim\limits_{z\to\infty}\ell_t(z)=\infty$. Accordingly, we use the one-sided boundary extension $\widetilde p_{At}(0,e_{Bt})=0$ for $e_{Bt}>0$ and $\widetilde p_{At}(e_{At},0)=1$ for $e_{At}>0$, with $\widetilde p_{At}(0,0)=1/2$ by convention.
    \item The log-odds function satisfies the \textbf{curvature condition}
    \begin{equation}\label{eq:curvature}
       0<\ell_t'(z)\le 1,
       \qquad
       \bigl|\ell_t''(z)\bigr|\le\ell_t'(z)\bigl(1-\ell_t'(z)\bigr),
       \qquad\forall z\in\mathbb R.
    \end{equation}
\end{enumerate}
\end{assumption}

Boundary regularity rules out ``headstart'' technologies with bounded log-odds, e.g., $\ell_t(z_t)=\log(e^{z_t}+C)$ with $C>0$ (i.e., $\widetilde{p}_{At}=\frac{e_{At}+Ce_{Bt}}{e_{At}+(1+C)e_{Bt}}$), under which a player retains a positive winning probability by choosing zero effort. The \textbf{curvature condition}~\eqref{eq:curvature} is a simple sufficient condition for the auxiliary two-player effort game to have a unique interior equilibrium at every pair of positive local stakes.\footnote{With $p\triangleq\widetilde p_{At}=1/(1+e^{-\ell_t(z)})$ and $z=\log(e_{At}/e_{Bt})$, one computes $\partial^2 \widetilde p_{At}/\partial e_{At}^2=\frac{p(1-p)}{e_{At}^2}\bigl[(1-2p)(\ell_t')^2+\ell_t''-\ell_t'\bigr]$ and $\partial^2 \widetilde p_{Bt}/\partial e_{Bt}^2=\frac{p(1-p)}{e_{Bt}^2}\bigl[(2p-1)(\ell_t')^2-\ell_t''-\ell_t'\bigr]$. The upper bound $\ell_t''\le\ell_t'(1-\ell_t')$ makes $\widetilde p_{At}$ concave in $e_{At}$, and the lower bound $\ell_t''\ge-\ell_t'(1-\ell_t')$ makes $\widetilde p_{Bt}$ concave in $e_{Bt}$; the first-order conditions then pin down a unique interior equilibrium.} \autoref{subsec:robust} discusses in detail the distinct roles played by the three parts of \autoref{ass:battle_technology}.

We introduce useful notation based on battle technologies. For each battle $t$ and each pair of positive local stakes $(x_A,x_B)$, consider the auxiliary one-battle contest in which player $A(t)$ chooses $e_A$, player $B(t)$ chooses $e_B$, and payoffs are $x_A\widetilde p_{At}(e_A,e_B)-c_{At}e_A$ and $x_B\bigl(1-\widetilde p_{At}(e_A,e_B)\bigr)-c_{Bt}e_B$, respectively. By \autoref{ass:battle_technology}, a unique interior equilibrium exists; let $\hat e_{At}(x_A,x_B)$ and $\hat e_{Bt}(x_A,x_B)$ denote its efforts, and define
\[
   p_{At}(x_A,x_B)
   \triangleq
   \widetilde p_{At}\bigl(\hat e_{At}(x_A,x_B),\hat e_{Bt}(x_A,x_B)\bigr),
   \qquad
   p_{Bt}(x_A,x_B)=1-p_{At}(x_A,x_B).
\]
We call $p_{At}$ the \emph{single-battle reduced-form success function}.

At one-sided stake boundaries, we complete the managers' reduced-form allocation game using
\[
   p_{At}(0,x_B)=0\quad(x_B>0),
   \qquad
   p_{At}(x_A,0)=1\quad(x_A>0),
\]
and set $p_{At}(0,0)=1/2$. These conventions define the managers' objectives on their entire strategy simplices. All equilibrium statements below refer to this completed allocation game. This qualification has no effect on equilibrium-path behavior because \autoref{lem:no_boundary} shows that both managers assign a positive reward to every battle in equilibrium.\footnote{The one-sided values are limiting conventions rather than outcomes of effort-subgame equilibria. For example, if $x_A>0=x_B$, player~$B$ chooses zero effort, while player~$A$ prefers an arbitrarily small positive effort to zero but has no optimal positive effort. The one-sided effort subgame therefore has no pure-strategy equilibrium. Moreover, the completion is not jointly continuous at $(0,0)$. Neither issue affects the analysis: an interior equilibrium candidate is constructed and shown to be globally optimal against an interior opponent, while boundary allocation profiles are ruled out separately in \autoref{lem:no_boundary}.}

\begin{lemma}[Reduced-form success function]\label{lem:csf}
For every battle $t$ and all $x_A,x_B>0$, $p_{At}$ is \textbf{homogeneous of degree zero} and strictly increasing in the ratio $x_A/x_B$. More precisely, writing $u_t\triangleq\log(x_A/x_B)$,
\begin{equation}\label{eq:csf}
  \log\frac{p_{At}(x_A,x_B)}{p_{Bt}(x_A,x_B)}=h_t(u_t),
  \qquad
  h_t(u_t)=\ell_t\!\left(u_t+\log\frac{c_{Bt}}{c_{At}}\right).
\end{equation}
Here, $h_t$ is strictly increasing and twice continuously differentiable.
\end{lemma}

The reduced-form success function is a \emph{local} object: it takes as inputs only the two stakes attached to battle~$t$ and that battle's primitive technology. Because common rescalings of $(x_A,x_B)$ leave $p_{At}$ unchanged, the reduced-form probability depends only on the relative incentives within the battle.

\paragraph{Leading example (power form).}
We illustrate every result with the canonical Tullock technology,
\begin{equation*}
  \widetilde{p}_{At}(e_{At},e_{Bt})=\frac{e_{At}^{\,r_t}}{e_{At}^{\,r_t}+e_{Bt}^{\,r_t}}, \qquad r_t\in(0,1],
\end{equation*}
where $r_t$ governs how sharply effort translates into wins. The log-odds function is linear, $\ell_t(z_t)=r_t z_t$ (so $\ell_t'\equiv r_t$ and $\ell_t''\equiv 0$), and the \textbf{curvature condition}~\eqref{eq:curvature} reduces to $r_t\le 1$, securing a unique interior equilibrium regardless of local stakes.\footnote{Within the power-form class, $r_t\leq1$ is the standard bound guaranteeing an interior pure-strategy equilibrium for every pair of positive local stakes. Thus, the curvature condition is tight for this leading example.} 
Specializing \autoref{lem:csf} then gives
\begin{equation}\label{eq:csf_tullock}
  p_{At}(v_{At},v_{Bt})=\frac{(c_{Bt}\,v_{At})^{r_t}}{(c_{Bt}\,v_{At})^{r_t}+(c_{At}\,v_{Bt})^{r_t}}, \qquad h_t(u_t)=r_t\,u_t+r_t\log\frac{c_{Bt}}{c_{At}}.
\end{equation}
The same analysis covers two generalizations: the \emph{biased} Tullock technology $\widetilde{p}_{At}=\frac{\zeta_t\,e_{At}^{\,r_t}}{\zeta_t\,e_{At}^{\,r_t}+(1-\zeta_t)\,e_{Bt}^{\,r_t}}$ with $\zeta_t\in(0,1)$, where $\zeta_t$ is an intrinsic bias ($\zeta_t=1/2$ is unbiased, $\zeta_t>1/2$ favors team~$A$); and technologies beyond the power class, such as $\ell_t(z_t)=\tfrac12 z_t+\tfrac{1}{10}\tanh z_t$, which is nonlinear and hence non-Tullock.

\subsection{Preliminaries}\label{subsec:preliminary}

So far the reduced form describes one battle in isolation. To embed it in the majoritarian contest, note that a battle matters only through its effect on the team-level outcome, and that this strategic importance enters both players' incentives through the same multiplier.

\begin{lemma}\label{lem:local_reduction}
Fix any precommitted allocation $(\mathbf{v}_A,\mathbf{v}_B)$ and any history under any fixed temporal structure.

(i) In a battle $t$ that can still affect the identity of the winning team, the increase in either team's continuation winning probability from winning rather than losing battle~$t$ is \textbf{common}, denoted by $\pi^{(t)}>0$.

(ii) If $v_{At},v_{Bt}>0$, then battle $t$ yields the same winning probability as the auxiliary contest with stakes $(\pi^{(t)}v_{At},\pi^{(t)}v_{Bt})$. Hence, by homogeneity of degree zero,
\[
   p_{At}\bigl(\pi^{(t)}v_{At},\pi^{(t)}v_{Bt}\bigr)
   =p_{At}(v_{At},v_{Bt}).
\]
\end{lemma}

In the baseline setting of simultaneous battles, $\pi^{(t)}>0$ holds for every battle. The economic content of \autoref{lem:local_reduction} is that the common multiplier $\pi^{(t)}$ cancels from each battle's \emph{winning probability}, so the team-level winning probability can be assembled directly from the local reduced forms $p_{it}(v_{At},v_{Bt})$.

\paragraph{The managers' problem.}
Manager~$i$ chooses $\mathbf{v}_i\in\Delta_i$ to maximize her team's chance of a majority. By \autoref{lem:local_reduction}, the managers' payoff-relevant battle probabilities are the single-battle reduced forms $p_{it}(v_{At},v_{Bt})$. Given these probabilities, conditional independence gives
\begin{equation}\label{eq:Probi}
\textnormal{Prob}_i(\mathbf{v}_A,\mathbf{v}_B)
\triangleq\sum_{\substack{w_i\subseteq\mathcal{N}\\|w_i|\ge N+1}}
\;\prod_{t\in w_i}p_{it}\;\prod_{q\in\mathcal{N}\setminus w_i}p_{jq},
\qquad j\in\{A,B\}\setminus\{i\},
\end{equation}
where $w_i$ ranges over the winning coalitions and each $p_{it}=p_{it}(v_{At},v_{Bt})$ is given by~\eqref{eq:csf} on the positive orthant and by the boundary conventions otherwise.

We have assumed simultaneous battles, but the managers' precommitted allocation problem is the same under any fixed timing. Consider a timing in which battles are grouped into clusters played in sequence and are simultaneous within each cluster; the fully simultaneous and fully sequential contests are the two extremes. Once a team has secured $N+1$ victories, later battles are not played and their players exert no effort. For ex ante evaluation, however, we may complete every stopped history by drawing the unplayed battles independently according to their reduced-form probabilities $(p_{At},p_{Bt})$. These artificial draws cannot change the identity of the team that has already secured a majority. They provide an outcome-equivalent completion under which team victory is represented by the same majority event as in the simultaneous contest.

\begin{lemma}[Same team-winning probability across all temporal structures]\label{lem:temporal}
For any fixed temporal structure and any precommitted allocation profile, team~$i$'s ex ante winning probability is given by~\eqref{eq:Probi}.
\end{lemma}

This invariance concerns precommitted reward allocations and ex ante team-winning probabilities. It does not imply that the same battles are actually played or that realized effort costs are invariant across temporal structures. The result need not survive history-contingent renegotiation, under which managers revise later rewards after observing earlier outcomes.

Finally, equilibrium allocations are always interior, so the finite log-odds representation~\eqref{eq:csf} applies throughout; boundary cases away from equilibrium are handled by the one-sided conventions above. Two degenerate cases can arise outside the interior: (i) one team concedes, $v_{it}=0<v_{jt}$, handing the battle to team~$j$; or (ii) both ignore it, $v_{At}=v_{Bt}=0$, leaving a coin flip. The next lemma rules out both.

\begin{lemma}\label{lem:no_boundary}
No pure-strategy Nash equilibrium $(\mathbf{v}_A^{*},\mathbf{v}_B^{*})$ has $v_{it}^{*}=0$ for any $i\in\{A,B\}$ and $t\in\mathcal{N}$.
\end{lemma}

The intuition is simple. If team~$i$ assigns zero reward to a battle that team~$j$ funds, team~$j$ already wins that battle with certainty. Team~$j$ can then transfer a sufficiently small amount from that battle to one it does not already win for sure without changing the conceded battle's outcome; positive pivotality makes the transfer profitable. If neither team funds a battle, either manager can commit an arbitrarily small reward to turn the coin flip into a sure local victory, securing a discrete gain at negligible cost.

\section{Equilibrium Characterization}
\label{Section:Equi}

We establish the existence and uniqueness of a pure-strategy Nash equilibrium in this section. By \autoref{lem:no_boundary}, we work on $\textnormal{int}(\Delta_A)\times\textnormal{int}(\Delta_B)$ throughout. The argument proceeds in three steps. First, using the pivotality decomposition in \autoref{lem:pivot}, we show that the two managers allocate their budgets in identical proportions across battles. Second, we construct an explicit candidate profile $(\hat{\mathbf{v}}_A,\hat{\mathbf{v}}_B)$ and verify that it is the unique solution to the first-order conditions (\autoref{thm:main}). Third, we prove that each team's contest-winning probability is log-concave in its own allocation given the opponent's (\autoref{thm:contprob}), making the first-order solution an equilibrium.

\subsection{Reward Schedule Alignment}

Let $\mathcal{N}_{-t}\triangleq\mathcal{N}\setminus\{t\}$ denote the set of all battles other than battle~$t$. Under majority rule, a battle matters for the contest outcome only when remaining battles are split evenly. This motivates the following definition.

\begin{definition}[Pivotality]\label{def:pivotality}
The \textbf{pivotality} of battle~$t$, denoted by $\theta(t)$, is defined as the probability that both teams win the same number of battles in $\mathcal{N}_{-t}$. It is the same for both teams.
\end{definition}

This definition gives the following marginal-decomposition identity.

\begin{lemma}\label{lem:pivot}
For any interior $(\mathbf{v}_A,\mathbf{v}_B)$, team $i$, and battle $t$,
\begin{equation*}
\frac{\partial\textnormal{Prob}_i}{\partial v_{it}}=\theta(t)\cdot\frac{\partial p_{it}}{\partial v_{it}}.
\end{equation*}
\end{lemma}

\begin{proof}
Since battles are independent conditional on $(\mathbf{v}_A,\mathbf{v}_B)$, we can condition on the outcomes of the battles in $\mathcal{N}_{-t}$ to decompose $\textnormal{Prob}_i=\Pr(\text{team $i$ wins at least $N+1$ battles in }\mathcal{N}_{-t})+\theta(t) p_{it}$. Only $p_{it}$ in the second term depends on $v_{it}$. Differentiating this expression with respect to $v_{it}$ gives the stated identity.
\end{proof}

We now establish a necessary condition for any pure-strategy equilibrium: the two managers allocate \emph{proportional} prizes across battles, meaning that the ratio of team~$B$'s prize to team~$A$'s prize is the same in every battle and equals the budget ratio $k\triangleq W_B/W_A$.

\begin{proposition}[Reward Schedule Alignment]\label{pro:PPP}
At any pure-strategy equilibrium,
\[
\frac{v_{Bt}^{*}}{v_{At}^{*}}=k
\qquad\text{for every }t\in\mathcal{N}.
\]
\end{proposition}

\begin{proof}
By \autoref{lem:no_boundary}, any equilibrium is interior, so the first-order conditions hold with finite multipliers $\lambda_A,\lambda_B>0$. By \autoref{lem:pivot}, for each battle $t$, 
\[
\theta(t)\,\frac{\partial p_{At}}{\partial v_{At}}=\lambda_A,\qquad \theta(t)\,\frac{\partial p_{Bt}}{\partial v_{Bt}}=\lambda_B.
\]
Dividing the two first-order conditions eliminates $\theta(t)$ and gives $\dfrac{\partial p_{At}/\partial v_{At}}{\partial p_{Bt}/\partial v_{Bt}}=\dfrac{\lambda_A}{\lambda_B}$ for all $t$.

Since \eqref{eq:csf} is homogeneous of degree zero in $(v_{At},v_{Bt})$, Euler's theorem gives $v_{At}\,\frac{\partial p_{At}}{\partial v_{At}}+v_{Bt}\,\frac{\partial p_{At}}{\partial v_{Bt}}=0$, while $p_{At}+p_{Bt}=1$ gives $\frac{\partial p_{At}}{\partial v_{Bt}}=-\frac{\partial p_{Bt}}{\partial v_{Bt}}$. Hence $v_{At}\,\frac{\partial p_{At}}{\partial v_{At}}=v_{Bt}\,\frac{\partial p_{Bt}}{\partial v_{Bt}}$, and combining with the ratio above gives $v_{Bt}/v_{At}=\lambda_A/\lambda_B$ for all $t$. The budget constraints then force $\lambda_A/\lambda_B=W_B/W_A=k$.
\end{proof}

This proportionality is notable because the two teams may face entirely different strategic considerations, yet they respond with identical relative allocations, differing only by the scalar $k$. The condition pins down the \emph{ratio} $v_{Bt}^{*}/v_{At}^{*}$ but leaves the \emph{levels} $v_{At}^{*}/W_A$ undetermined; those are characterized in \autoref{subsec:main}.

Consider a three-battle example in which team~$A$ is slightly stronger in battles 1 and 2 and much weaker in battle 3. One might expect $A$ to concentrate rewards on the first two battles and concede the third. \autoref{pro:PPP} rules out such asymmetric tilting. Homogeneity of degree zero implies that only the within-battle reward ratio affects the local outcome, while pivotality is common to the two teams. Combining these two facts with the managers' first-order conditions and Euler's identity forces the normalized allocations to coincide. This is an equilibrium alignment result, not a sequential-response argument: neither manager observes and then offsets the other's deviation.

When $W_A=W_B$, proportionality becomes equality, $v_{At}^*=v_{Bt}^*$ in every battle. Cost and technological asymmetries still matter because they change equilibrium battle probabilities, closeness, pivotality, and therefore the common reward assigned to each battle. What they do not generate is disagreement between the managers about relative reward shares.

\subsection{Main Results}\label{subsec:main}

The alignment property immediately determines the battle-winning probabilities in any pure-strategy equilibrium. If $v_{Bt}^{*}=k v_{At}^{*}$, then the log prize ratio $u_t^{*}=-\log k$ is the same across all battles. Define the resulting \emph{proportional ratio} probabilities as
\begin{equation}\label{eq:pStar_general}
p_{At}^{*}\triangleq p_{At}(1,k)
=\frac{1}{1+e^{-h_t(-\log k)}},
\qquad
p_{Bt}^{*}=1-p_{At}^{*}.
\end{equation}
By \autoref{pro:PPP}, they are the equilibrium probabilities and depend only on the primitives of the corresponding battle and the budget ratio~$k$. For the leading example, $h_t(u)=r_t u+r_t\log(c_{Bt}/c_{At})$ gives the explicit formula
\begin{equation}\label{eq:pStar}
p_{At}^{*}=\frac{c_t}{c_t+k^{r_t}},\qquad p_{Bt}^{*}=\frac{k^{r_t}}{c_t+k^{r_t}},
\end{equation}
where $c_t\triangleq\rho_t^{r_t}$, with $\rho_t\triangleq c_{Bt}/c_{At}$, captures player~$A(t)$'s cost advantage. The corresponding \emph{pivotality} is
\begin{equation}\label{eq:eq-pivot}
\theta^{*}(t)=\sum_{\substack{w_i\subseteq\mathcal{N}_{-t}\\|w_i|=N}}
\;\prod_{s\in w_i}p^*_{is}\;\prod_{q\in\mathcal{N}_{-t}\setminus w_i}p^*_{jq},
\qquad j\in\{A,B\}\setminus\{i\}.
\end{equation}
Since every $p_{it}^{*}\in(0,1)$, each $\theta^{*}(t)$ is strictly positive.

To understand how managers split their budgets, consider team~$A$'s problem. For any technology that is homogeneous of degree zero, abbreviated HD-0, differentiating the reduced form~\eqref{eq:csf} and using $\eta_t\triangleq h_t(u_t)$ shows that the marginal return to prize investment factors cleanly: $\frac{\partial p_{At}}{\partial v_{At}}=\frac{p_{At}\,p_{Bt}\,h_t'(u_t)}{v_{At}}$.\footnote{Write $\eta_t=\log\frac{p_{At}}{1-p_{At}}=h_t(u_t)$ with $u_t=\log v_{At}-\log v_{Bt}$, and differentiate along $v_{At}\mapsto u_t\mapsto\eta_t\mapsto p_{At}$. Since $v_{Bt}$ is fixed, $\partial u_t/\partial v_{At}=1/v_{At}$, so $\partial\eta_t/\partial v_{At}=h_t'(u_t)/v_{At}$. Inverting the log-odds gives the logistic $p_{At}=e^{\eta_t}/(1+e^{\eta_t})$, whose derivative is $\partial p_{At}/\partial\eta_t=p_{At}(1-p_{At})=p_{At}p_{Bt}$. The chain rule then yields this formula.}
Evaluating at the proportional ratio $v_{Bt}=k\,v_{At}$ (where $u_t=-\log k$) yields
\begin{equation}\label{eq:dpAt}
\frac{\partial p_{At}}{\partial v_{At}}\bigg|_{v_{Bt}=k\,v_{At}}=\frac{h_t'(-\log k)\,p^*_{At}\,p^*_{Bt}}{v_{At}}.
\end{equation}
Since the denominator $v_{At}$ is merely a budget-scaling factor, the battle-level marginal impact is governed by the numerator $h_t'(-\log k)\,p^*_{At}\,p^*_{Bt}$, which we single out as a measure of each battle's importance. Team~$B$'s problem yields the \emph{same} numerator, since $v_{At}\,\frac{\partial p_{At}}{\partial v_{At}}=v_{Bt}\,\frac{\partial p_{Bt}}{\partial v_{Bt}}$. Moreover, $h_t'(-\log k)=\ell_t'\!\left(-\log k+\log\frac{c_{Bt}}{c_{At}}\right)$. For the leading example, $h_t'\equiv\ell_t'\equiv r_t$, so \eqref{eq:dpAt} reduces to $r_t\,p^*_{At}\,p^*_{Bt}/v_{At}$.

\begin{definition}[Responsiveness]\label{def:responsiveness}
The \textbf{responsiveness} of battle~$t$ at the proportional ratio is defined as
\begin{equation*}
    R_t \triangleq \underbrace{h_t'(-\log k)}_{\text{discriminatory power}}\cdot\underbrace{p_{At}^*\,p_{Bt}^*}_{\text{closeness}},\qquad\forall t.
\end{equation*}
For the leading example $h_t'\equiv r_t$, so $R_t=r_t\,p_{At}^*p_{Bt}^*$.
\end{definition}

The quantity $R_t$ measures how effectively prize investment translates into a higher winning probability in battle~$t$ at the candidate proportional ratio. It is the product of two components. The first is the \emph{(local) discriminatory power} $h_t'(-\log k)$, which measures the slope of the log-odds at the proportional prize ratio and thus how sharply the technology distinguishes between unequal efforts. In the Tullock framework, it equals $r_t$, consistent with the usual meaning of discriminatory power in contest theory. The second is the \emph{closeness} $p_{At}^*p_{Bt}^*$, which measures how evenly matched the two sides are. Intuitively, responsiveness increases with both discriminatory power and closeness. Closeness is maximized at $1/4$ when $p_{At}^*=p_{Bt}^*=1/2$. Conversely, $R_t$ vanishes as either side's effective strength dominates, because a lopsided battle whose outcome is nearly predetermined cannot be swayed by marginal changes in prize allocation.

Under majority rule, however, the marginal return to prize investment depends not only on the responsiveness of the corresponding battle, but also on the likelihood that winning that battle changes the overall contest outcome. We therefore introduce a composite measure that captures both dimensions.

\begin{definition}[Salience]\label{def:salience}
The \textbf{salience} of battle~$t$ at the proportional ratio is defined as $S_t\triangleq\theta^{*}(t)\cdot R_t$.
\end{definition}

Note that salience naturally combines two forces: how \emph{responsive} a battle is to additional investment and how \emph{pivotal} it is for the contest outcome. Since $p^*_{At},p^*_{Bt}\in(0,1)$, $h_t'>0$ by \autoref{ass:battle_technology}, and $\theta^*(t)>0$, salience is strictly positive for every battle. With this measure in hand, we can state the main characterization.

\begin{theorem}[Reward Schedule Alignment and Salience]\label{thm:main}
Under \autoref{ass:battle_technology}, for any odd number $2N+1$ of battles with $N\geq1$ and any budgets $W_A,W_B>0$, the completed allocation game has a \textbf{unique} pure-strategy Nash equilibrium $(\mathbf{v}_A^{*},\mathbf{v}_B^{*})$, given by
\[
v_{At}^{*}=W_A\cdot\frac{S_t}{\sum_{\tau}S_{\tau}},\qquad v_{Bt}^{*}=W_B\cdot\frac{S_t}{\sum_{\tau}S_{\tau}},\qquad \forall t.
\]
Consequently, the probabilities in \eqref{eq:pStar_general} and the pivotalities in \eqref{eq:eq-pivot} are their equilibrium values.
\end{theorem}

The equilibrium can be computed directly from the primitives, without solving a fixed-point problem:
\begin{enumerate}
    \item For each battle~$t$, compute the proportional-ratio winning probabilities $p_{At}^*=p_{At}(1,k)$ and $p_{Bt}^*=1-p_{At}^*$.
    \item Compute $\theta^*(t)$, the probability that the remaining $2N$ battles split evenly, using~\eqref{eq:eq-pivot}. All pivotalities can be computed jointly in quadratic time.\footnote{Let $n\triangleq2N+1$ and $P(x)\triangleq\prod_{s\in\mathcal N}(p_{Bs}^*+p_{As}^*x)$. For each battle~$t$, $\theta^*(t)$ is the coefficient on $x^N$ in $P(x)/(p_{Bt}^*+p_{At}^*x)$. The coefficients of $P$ can be constructed by the standard Bernoulli recursion, after which synthetic division by each linear factor recovers the corresponding pivotality in $O(N)$ operations. Thus all pivotalities can be obtained in $O(nN)=O(n^2)$ operations.}
    \item Form each battle's salience and normalize across battles:
    \[
        S_t
        =
        h_t'(-\log k)\,p_{At}^*p_{Bt}^*\theta^*(t),
        \qquad
        v_{it}^*
        =
        W_i\frac{S_t}{\sum_{\tau\in\mathcal N}S_\tau},
        \quad i\in\{A,B\}.
    \]
\end{enumerate}

Each manager allocates budget across battles in proportion to battle salience, directing larger shares to battles that are simultaneously more responsive and more pivotal. Neither component alone governs the allocation: a highly pivotal battle may receive a small share if it is too lopsided for additional incentives to make a difference, while a perfectly balanced battle may receive little if it is unlikely to be decisive---for instance, in a three-battle contest, a balanced battle carries low pivotality when one team dominates in the other two battles.

\begin{remark}
The salience $S_t$ is a battle weight generated by decentralized competition between managers. \citet*{Feng-2024-Optimal} obtain a different weight in a centralized design problem: their organizer can jointly choose battle scores and a headstart to maximize total effort, which can make weight increase with asymmetry. Here managers take the aggregation rule as given and maximize their own winning probabilities. A highly asymmetric battle is less responsive to marginal reward changes, so the equilibrium assigns it a smaller common weight, all else equal. The contrast reflects the institutional difference between centralized rule design and strategic reward allocation.
\end{remark}

\begin{remark}
The closeness--pivotality component of salience is related to the electoral-targeting index in \citet*{stromberg2008electoral}: competing campaigns devote more attention to states that are both contestable and potentially decisive. Our object is different. Managers allocate victory-contingent rewards to strategic agents rather than deploy campaign resources directly, and local \textbf{discriminatory power} enters the allocation rule in addition to closeness and pivotality.
\end{remark}

\begin{remark}
Our results extend the temporal-structure independence of team contests established by \citet*{Fu-2015-Team} to the dividing-the-spoils game played by competing team managers. By \autoref{lem:temporal}, the equilibrium $(\mathbf{v}_A^*,\mathbf{v}_B^*)$ and the resulting team-winning probabilities are invariant to the temporal structure. In a fully sequential setting where each outcome is publicly revealed before the next battle begins, one might expect managers to exploit the order of play, loading prizes onto early battles to build momentum or reserving resources for a clinching scenario. No such dynamic consideration arises under precommitment. The unordered product formula in \eqref{eq:Probi} already integrates over all histories: later battles matter only in those histories in which the majority has not yet been clinched, and the probability of reaching such histories is itself pinned down by the same precommitted battle probabilities. The temporal structure therefore changes the order in which uncertainty is resolved, not the ex ante allocation problem.
\end{remark}

In the leading example the salience is $S_t=\theta^{*}(t)\,r_t\,p^{*}_{At}p^{*}_{Bt}$, recovering the explicit Tullock characterization. The biased Tullock technology, even with a battle-specific bias $\zeta_t\in(0,1)$, is also covered: $\zeta_t$ enters the reduced form only through the intercept of $h_t$, so the formula applies verbatim. The bias shifts the win probabilities $p^{*}_{At}$, and hence the \emph{level} of each salience, but preserves the alignment property.

Consider a three-battle contest with equal budgets $W_A = W_B = 1$, common Tullock discriminatory power $r_t = 1$ for all $t$, and cost ratios $c_1=1$, $c_2 = 4$, $c_3 = 2$. Team~$A$'s player is more able in battles 2 and 3, and the two sides are evenly matched in battle 1. Since $k = 1$, the equilibrium battle-winning probabilities~\eqref{eq:pStar} reduce to $p^*_{At} = c_t/(1+c_t)$. To illustrate the second step of the computation, battle~1's pivotality is the probability that battles~2 and~3 split evenly: $\theta^{*}(1)=p^{*}_{A2}\,p^{*}_{B3}+p^{*}_{B2}\,p^{*}_{A3}=\tfrac{4}{5}\cdot\tfrac{1}{3}+\tfrac{1}{5}\cdot\tfrac{2}{3}=\tfrac{2}{5}$; the other pivotalities follow in the same way. \autoref{tab:example} reports the resulting equilibrium computed from $S_t=\theta^{*}(t)\,r_t\,p^{*}_{At}p^{*}_{Bt}$.

\begin{table}[!htbp]
\centering
\caption{Equilibrium in the three-battle example.}
\label{tab:example}
\medskip
\begin{tabular}{ccccccc}
\hline
$t$ & $c_t$ & $p^*_{At}$ & $\theta^*(t)$ & $R_t=r_t\,p^{*}_{At}p^{*}_{Bt}$ & $S_t=\theta^*(t)R_t$ & $v^*_{At} = v^*_{Bt}$ \\
\hline
1 & 1 & $1/2$  & $2/5$  & $1/4 = 0.250$ & $1/10=0.100$ & $45/131 \approx 0.343$ \\
2 & 4 & $4/5$  & $1/2$  & $4/25 = 0.160$ & $2/25=0.080$ & $36/131 \approx 0.275$ \\
3 & 2 & $2/3$  & $1/2$  & $2/9 \approx 0.222$ & $1/9\approx0.111$ & $50/131 \approx 0.382$ \\
\hline
\end{tabular}
\end{table}

Several features of the equilibrium are worth noting. First, battle~3 receives the largest prize share (0.382), even though it is neither the most nor the least balanced battle. It earns the largest share because it combines high pivotality ($\theta^*(3) = 1/2$, since the other two battles split with substantial probability) with intermediate responsiveness ($R_3=2/9$). Second, battle~1 has the highest responsiveness ($R_1=1/4$), owing to its perfect balance, yet receives a smaller share than battle~3 because its pivotality is lower ($\theta^*(1) = 2/5$, depressed by team~$A$'s advantages in the other two battles). Third, battles~2 and~3 share the same high pivotality ($1/2$) but differ in responsiveness ($4/25$ versus $2/9$), illustrating that pivotality alone does not determine budget shares.

One natural special case is worth recording.

\begin{corollary}\label{cor:sameC}
If every battle is \textbf{effectively balanced} at the equilibrium ratio, i.e.\ $p_{At}^*=p_{Bt}^*=1/2$ (equivalently $h_t(-\log k)=0$) for all $t\in\mathcal{N}$, then both the closeness term and the pivotality term are constant across battles, so the equilibrium allocation is governed solely by the local discriminatory powers: $v_{it}^{*}=W_i\cdot h_t'(-\log k)/\sum_{\tau}h_\tau'(-\log k)$. For the leading example this is the case $c_t=k^{r_t}$, giving $v_{it}^{*}=W_i\cdot r_t/\sum_{\tau}r_{\tau}$.
\end{corollary}

When each battle is effectively balanced, the more discriminatory the battle, the larger the prize share it commands. In the equal-budget case ($k=1$), the leading-example condition $c_t=k^{r_t}$ reduces to $c_t=1$.

\subsection{Proof of \autoref{thm:main}}

The following condition plays a key role in the proof; it is established in \autoref{subsec:logconcavity}.

\paragraph{Log-concavity condition.}
Under \autoref{ass:battle_technology}, for $i\in\{A,B\}$ and any fixed $\mathbf{v}_{j}\in\textnormal{int}(\Delta_j)$, $\mathbf{v}_i\mapsto\textnormal{Prob}_i$ is log-concave on $\textnormal{int}(\Delta_i)$.

\begin{remark}\label{rmk:nocomposition}
One might hope to establish quasiconcavity of $\textnormal{Prob}_i$ in $\mathbf{v}_i$ by a composition argument: each $v_{it}\mapsto p_{it}(v_{it},\cdot)$ is concave, and if $\textnormal{Prob}_i$ were quasiconcave in the vector of battle-winning probabilities $\mathbf{p}_i$, quasiconcavity in $\mathbf{v}_i$ would follow. However, $\textnormal{Prob}_i$ is \textbf{not} quasiconcave in $\mathbf{p}_i$, even in the three-battle case. To see this, take $\mathbf{p}_A=(0.999,0.999,0.001)$ and $\mathbf{p}_A'=(0.001,0.999,0.999)$: both give $\textnormal{Prob}_A\approx 0.998$, but their midpoint $(0.5,0.999,0.5)$ gives $\textnormal{Prob}_A\approx 0.7495$, violating quasiconcavity. Nor is $\textnormal{Prob}_i$ generally concave in $\mathbf{v}_i$.
\end{remark}

We now prove \autoref{thm:main}. Throughout, suppose \autoref{ass:battle_technology} holds.

\medskip
\noindent\textbf{Existence.}
Define the candidate allocations $\hat v_{At}\triangleq W_A\cdot S_t/\sum_{\tau}S_{\tau}$ and $\hat v_{Bt}\triangleq W_B\cdot S_t/\sum_{\tau}S_{\tau}$ for $t\in\mathcal{N}$. Since each $S_t>0$, the profile $(\hat{\mathbf{v}}_A,\hat{\mathbf{v}}_B)\in\textnormal{int}(\Delta_A)\times\textnormal{int}(\Delta_B)$. Note that $\hat v_{Bt}=k\,\hat v_{At}$ where $k=W_B/W_A$. We now verify that the candidate satisfies the first-order conditions. Given \eqref{eq:dpAt}, the Lagrangian first-order condition $\theta(t)\cdot\partial p_{At}/\partial v_{At}=\lambda_A$ becomes $S_t/\hat v_{At}=\lambda_A$ for each $t\in\mathcal{N}$. Since $\hat v_{At}=W_A\cdot S_t/\sum_{\tau}S_{\tau}$, this gives $\lambda_A=\sum_{\tau}S_{\tau}/W_A$ for every $t$. Hence $\hat{\mathbf{v}}_A$ satisfies the Lagrangian first-order condition for team $A$ given the opponent's allocation $\hat{\mathbf{v}}_B$. For team $B$, the symmetric derivative is $\partial p_{Bt}/\partial v_{Bt}=h_t'(-\log k)\,p_{At}^{*}p_{Bt}^{*}/v_{Bt}$ at the same proportional ratio, so $\hat v_{Bt}=W_B\cdot S_t/\sum_{\tau}S_{\tau}$ yields $\lambda_B=\sum_{\tau}S_{\tau}/W_B$.

Suppose the \textbf{log-concavity condition} holds, so that $\log\textnormal{Prob}_A(\cdot,\hat{\mathbf{v}}_B)$ is concave on $\textnormal{int}(\Delta_A)$. Because $\log$ is strictly increasing and $\textnormal{Prob}_A>0$ on the interior, the stationary points of $\textnormal{Prob}_A(\cdot,\hat{\mathbf{v}}_B)$ and of $\log\textnormal{Prob}_A(\cdot,\hat{\mathbf{v}}_B)$ coincide. Hence, any point satisfying the constrained first-order conditions for $\textnormal{Prob}_A$ is a critical point of the concave function $\log\textnormal{Prob}_A$, and therefore a global maximizer of $\log\textnormal{Prob}_A$---equivalently, of $\textnormal{Prob}_A$. Since $\hat{\mathbf{v}}_A$ satisfies the Lagrangian first-order conditions for team~$A$ given $\hat{\mathbf{v}}_B$, it follows that $\hat{\mathbf{v}}_A$ is a global maximizer of $\textnormal{Prob}_A(\cdot,\hat{\mathbf{v}}_B)$ on $\textnormal{int}(\Delta_A)$.

We then show that no boundary strategy $\mathbf{v}'_A\in\Delta_A\setminus\textnormal{int}(\Delta_A)$ yields a higher payoff. Since $\hat{\mathbf{v}}_B\in\textnormal{int}(\Delta_B)$, the opponent assigns $\hat v_{Bt}>0$ to every battle. By the boundary regularity in \autoref{ass:battle_technology}, $h_t(u_t)=\ell_t\bigl(u_t+\log(c_{Bt}/c_{At})\bigr)\to-\infty$ as $u_t\to-\infty$, so the reduced form $p_{At}(\cdot\,,\hat v_{Bt})$ extends continuously to $v_{At}=0$ with $p_{At}(0,\hat v_{Bt})=0$, matching the one-sided boundary extension; hence it is continuous on $[0,W_A]$; for the leading example this is immediate from \eqref{eq:csf_tullock}, whose denominator is bounded below by $(c_{At}\hat v_{Bt})^{r_t}>0$. Consequently, $\textnormal{Prob}_A(\cdot\,,\hat{\mathbf{v}}_B)$ is continuous on the entire closed simplex $\Delta_A$. Now take any $\mathbf{v}'_A\in\Delta_A\setminus\textnormal{int}(\Delta_A)$ and choose a sequence $\mathbf{v}_A^{(n)}\in\textnormal{int}(\Delta_A)$ with $\mathbf{v}_A^{(n)}\to\mathbf{v}'_A$. By global optimality of $\hat{\mathbf{v}}_A$, $\textnormal{Prob}_A(\mathbf{v}_A^{(n)},\hat{\mathbf{v}}_B)\le\textnormal{Prob}_A(\hat{\mathbf{v}}_A,\hat{\mathbf{v}}_B)$ for every $n$. By continuity, $\textnormal{Prob}_A(\mathbf{v}'_A,\hat{\mathbf{v}}_B)=\lim_n\textnormal{Prob}_A(\mathbf{v}_A^{(n)},\hat{\mathbf{v}}_B)\le\textnormal{Prob}_A(\hat{\mathbf{v}}_A,\hat{\mathbf{v}}_B)$.

The preceding argument shows that $\hat{\mathbf{v}}_A$ is a global maximizer of $\textnormal{Prob}_A(\cdot,\hat{\mathbf{v}}_B)$ on $\Delta_A$. By the team-$B$ version of the \textbf{log-concavity condition}, the function $\log\textnormal{Prob}_B(\hat{\mathbf{v}}_A,\cdot)$ is concave on $\textnormal{int}(\Delta_B)$. An identical argument applied to team~$B$ shows that $\hat{\mathbf{v}}_B$ is a global maximizer of $\textnormal{Prob}_B(\hat{\mathbf{v}}_A,\cdot)$ on $\Delta_B$. Therefore $(\hat{\mathbf{v}}_A,\hat{\mathbf{v}}_B)$ is a pure-strategy Nash equilibrium.

\bigskip
\noindent\textbf{Uniqueness.}
Let $(\mathbf{v}_A^{*},\mathbf{v}_B^{*})$ be any pure-strategy Nash equilibrium. By \autoref{lem:no_boundary}, it is interior; by \autoref{pro:PPP}, $v_{Bt}^{*}=k\,v_{At}^{*}$ for all $t$. The proportional-prize condition pins down every battle-winning probability as a function of primitives alone, so the Lagrangian first-order conditions reduce to $v_{At}^{*}\cdot\lambda_A=S_t$ for each $t$, where $S_t$ depends only on primitives and $k$. Summing over $t$ and applying the budget constraint determines $\lambda_A$ and hence each $v_{At}^{*}$ uniquely. Any equilibrium must therefore coincide with the candidate $(\hat{\mathbf{v}}_A,\hat{\mathbf{v}}_B)$.

\subsection{Proof of Log-concavity Condition}\label{subsec:logconcavity}

This section is devoted to proving the following result.

\begin{theorem}\label{thm:contprob}
Under \autoref{ass:battle_technology}, the \textbf{log-concavity condition} holds for both $i\in\{A,B\}$ and every $N \geq 1$.
\end{theorem}

Recall that the model and the equilibrium characterization in \autoref{thm:main} allow an arbitrary odd number $2N+1$ of battles. Establishing the log-concavity condition for every $N\ge1$ is therefore the central technical task of the paper.

The Hessian of a manager's log team-winning probability reflects two forces. The first is the own-battle curvature generated when an additional reward changes the local winning probability. The curvature condition in \autoref{ass:battle_technology} makes this component sufficiently negative. The second is the dependence created by conditioning on team victory: although battle outcomes are independent before conditioning, they become correlated within the event that the team wins a majority. The main question is whether this conditional dependence can overturn the own-battle curvature. We show that it cannot under majority rule.

The proof proceeds in three steps. First, \autoref{lemma:Hessian} expresses the Hessian condition as positive semidefiniteness of a matrix $M$ that combines own-battle curvature with conditional covariance. Second, the curvature condition bounds $M$ below by a technology-free matrix $M^0$. Third, \autoref{prop:covdom_all} proves $M^0\succeq0$ by converting it into a polynomial matrix and showing that every coefficient matrix is positive semidefinite. We present the argument for team~$A$ and then explain how the lower half of the curvature condition yields the corresponding result for team~$B$.

\medskip

We now start the proof of \autoref{thm:contprob}; throughout, suppose \autoref{ass:battle_technology} holds. For each $t\in\mathcal{N}=\{1, \dots, 2N+1\}$, let $X_t \sim \textnormal{Bernoulli}(p_t)$ be independent indicators of team $A$ winning battle $t$, where $p_t = p_{At}(v_{At}, v_{Bt})$ is the current battle-winning probability. Let
\[
A \triangleq \Bigl\{ \mathbf{x} \in \{0,1\}^{2N+1} : \sum_{t} x_t \geq N+1 \Bigr\}
\]
denote the majority event, and write $f(\mathbf{v}_A) \triangleq \textnormal{Prob}_A(\mathbf{v}_A, \mathbf{v}_B) = \Pr(A)$ for team $A$'s probability of winning the overall contest. Our goal is to show that $\log f$ is concave in $\mathbf{v}_A$.

\paragraph{From log-concavity to a matrix condition.}
Log-concavity of $f$ in $\mathbf{v}_A$ is, by definition, the negative semidefiniteness of the Hessian $H$ of $\log f$. Define
\[
G_t \triangleq \mathbb{E}[X_t \mid A] - p_t,
\]
which measures how much conditioning on team $A$'s overall victory raises the marginal probability of winning battle $t$. Let $D = \mathrm{diag}(p_t(1-p_t))$ collect the unconditional Bernoulli variances, let $\Sigma_A = \operatorname{Cov}(\mathbf{X} \mid A)$ denote the conditional covariance matrix, and set
\[
M \triangleq D - \Sigma_A + \mathrm{diag}\!\Bigl(G_t\,\tfrac{h_t'-h_t''}{(h_t')^2}\Bigr),
\]
with $p_t=p_{At}(v_{At},v_{Bt})$ and each $h_t',h_t''$ evaluated at the current log prize-ratio $u_t=\log(v_{At}/v_{Bt})$ (the point at which concavity is checked). For the leading example $h_t'\equiv r_t$, $h_t''\equiv 0$, so the diagonal correction is $G_t/r_t$ and $M=D-\Sigma_A+\mathrm{diag}(G_t/r_t)$.

\begin{lemma}[Negative semidefiniteness of $H$]\label{lemma:Hessian}
$H \preceq 0$ if and only if $M \succeq 0$.
\end{lemma}

The proof rests on a convenient reparametrization. Working in the log-odds $\eta_t = \log\bigl(p_t / (1-p_t)\bigr)$, the joint distribution of $\mathbf{X}$ becomes a canonical exponential family, and standard exponential-family identities yield
\[
\frac{\partial \log f}{\partial \eta_t} = G_t,
\qquad
\frac{\partial^2 \log f}{\partial \eta_t \, \partial \eta_s} = \operatorname{Cov}(X_t, X_s \mid A) - \delta_{ts}\, p_t(1-p_t).
\]
The first identity has a clean interpretation: raising the log-odds of winning battle $t$ improves team $A$'s overall win probability at a rate proportional to $G_t$. The second identity says that the curvature of $\log f$ along $\boldsymbol{\eta}$ is governed by how battle outcomes co-move conditional on victory, net of each Bernoulli's intrinsic variance $p_t(1-p_t)$; the latter contributes only on the diagonal, as recorded by the Kronecker delta $\delta_{ts}$.

To translate these derivatives back into the original prize variables, we use the HD-0 structure $\eta_t = h_t(u_t)$ with $u_t=\log v_{At}-\log v_{Bt}$, which gives $\partial \eta_t / \partial v_{At} = h_t'/v_{At}$ and $\partial^2 \eta_t / \partial v_{At}^2 = (h_t''-h_t')/v_{At}^2$. Setting $A_d = \mathrm{diag}(h_t'/v_{At})$, the chain rule produces
\[
H = -A_d(D-\Sigma_A)A_d - \mathrm{diag}\!\Bigl(G_t\,\tfrac{h_t'-h_t''}{v_{At}^2}\Bigr).
\]
The change of variables $\widetilde{\mathbf{u}} = A_d \mathbf{u}$---invertible since $h_t'>0$ by \autoref{ass:battle_technology}---then converts the quadratic form $\mathbf{u}^\top H \mathbf{u}$ into $-\widetilde{\mathbf{u}}^\top M \, \widetilde{\mathbf{u}}$, so that $H \preceq 0$ if and only if $M \succeq 0$. For the leading example, $A_d=RV^{-1}$ with $R=\mathrm{diag}(r_t)$, $V=\mathrm{diag}(v_{At})$, recovering $H=(RV^{-1})(\Sigma_A-D)(RV^{-1})-RV^{-2}\mathrm{diag}(G_t)$.

\paragraph{Reduction to a technology-free matrix.}
The matrix $M$ is still inconvenient to analyze directly because its diagonal correction depends on the technology through $h_t',h_t''$. To eliminate this dependence, we bound $M$ below by a simpler, technology-free matrix. The key observation is that $G_t \geq 0$: conditioning on team $A$'s overall victory can only raise the marginal probability of winning any individual battle. Intuitively, the event $A$ is a ``good news'' event for every $X_t$, since it requires that at least $N+1$ of the indicators turn up favorable.\footnote{Formally, by Bayes' rule and the independence of $X_t$ from $\mathbf{X}_{-t}$, $\mathbb{E}[X_t \mid A] = \beta p_t / [\beta p_t + \gamma(1-p_t)]$, where $\beta = \Pr(\sum_{s \neq t} X_s \geq N)$ and $\gamma = \Pr(\sum_{s \neq t} X_s \geq N+1) \leq \beta$. Replacing $\gamma$ by $\beta$ enlarges the denominator and yields $\mathbb{E}[X_t \mid A] \geq p_t$.\label{footnote}} The relevant ingredient from the \textbf{curvature condition}~\eqref{eq:curvature} is its \emph{upper} bound $\ell_t''\le\ell_t'(1-\ell_t')$: since $h_t(u)=\ell_t\bigl(u+\log(c_{Bt}/c_{At})\bigr)$ shares slope and curvature with $\ell_t$, this is equivalent to $h_t'-h_t''\ge(h_t')^2$, i.e.\ $(h_t'-h_t'')/(h_t')^2 \geq 1$. Combined with $G_t\ge0$, this yields $G_t(h_t'-h_t'')/(h_t')^2 \geq G_t$, so the diagonal of $M$ pointwise dominates that of the technology-free matrix
\[
M^{0} \triangleq D - \Sigma_A + \mathrm{diag}(G_t) ,
\]
and consequently $M^0 \preceq M$. It therefore suffices to establish $M^0 \succeq 0$, which is our second key lemma. In the leading example, the \textbf{curvature condition} is $r_t\le 1$, and the relaxation is the elementary $G_t/r_t\ge G_t$.

\begin{lemma}[Positive semidefiniteness of $M^0$]\label{prop:covdom_all}
$M^0 \succeq 0$.
\end{lemma}

We sketch the proof strategy here. The first step is a change of variables to the loss odds $\phi_t = (1-p_t)/p_t \in (0, \infty)$, under which the conditional distribution given the majority event becomes clean to manipulate: every entry of $M^0$ admits a closed form as a ratio of polynomials in $\boldsymbol{\phi}$. Clearing denominators reduces the question to showing that a polynomial matrix $\widetilde M(\boldsymbol{\phi})$ is positive semidefinite (PSD) for all positive $\boldsymbol{\phi}$. We then expand $\widetilde M$ in monomials and establish the stronger statement that \emph{each coefficient matrix is individually PSD}. Symmetry of the construction forces every such coefficient matrix into a simple two-block form parametrized by just five numbers, which we compute explicitly as binomial coefficients. Finally, we verify that these five-number matrices are PSD by reducing them to a handful of scalar inequalities and proving those inequalities algebraically.

\medskip
Together, \autoref{lemma:Hessian} reduces log-concavity of $f$ to the matrix inequality $M \succeq 0$, the technology-free relaxation reduces $M \succeq 0$ to $M^0 \succeq 0$, and \autoref{prop:covdom_all} delivers the latter. Combining these three steps establishes that $\log f$ is concave in $\mathbf{v}_A$.

\paragraph{The team-$B$ side.}
The argument above establishes the \textbf{log-concavity condition} for $i=A$. The case $i=B$ follows by an identical construction with the roles of the two teams interchanged; we record why it holds under the \emph{same} \autoref{ass:battle_technology}, since the model is not symmetric in primitives and $\ell_t$ need not be anonymous. Let $Y_t=\mathbf 1\{B\text{ wins battle }t\}\sim\textnormal{Bernoulli}(q_t)$ with $q_t=p_{Bt}=1-p_t$, let the majority event for $B$ be $B\triangleq\{\mathbf y\in\{0,1\}^{2N+1}:\sum_t y_t\ge N+1\}$, and write $f_B\triangleq\textnormal{Prob}_B$. Repeating the reparametrization with the team-$B$ log-odds
\[
\eta_t^{B}=\log\frac{q_t}{1-q_t}=-h_t(u_t),\qquad u_t=\log\frac{v_{At}}{v_{Bt}},
\]
and differentiating in $v_{Bt}$ (so that $\partial u_t/\partial v_{Bt}=-1/v_{Bt}$) reproduces \autoref{lemma:Hessian} verbatim, with the diagonal correction now governed by $(h_t'+h_t'')/(h_t')^2$ in place of $(h_t'-h_t'')/(h_t')^2$. The \emph{lower} bound in \eqref{eq:curvature}, $\ell_t''\ge-\ell_t'(1-\ell_t')$---equivalently $h_t'+h_t''\ge(h_t')^2$, i.e.\ $(h_t'+h_t'')/(h_t')^2\ge1$---combined with $G_t^{B}\triangleq\mathbb E[Y_t\mid B]-q_t\ge0$, again yields the technology-free relaxation to $M^0$, and \autoref{prop:covdom_all} delivers $M^0\succeq0$. Hence $\log f_B$ is concave in $\mathbf v_B$. The two-sided \textbf{curvature condition}~\eqref{eq:curvature} thus secures log-concavity for \emph{both} teams: the upper bound controls team $A$ and the lower bound controls team $B$. In the leading example $h_t''\equiv0$, so both bounds collapse to $r_t\le 1$ and the two sides are identical. This completes the proof of \autoref{thm:contprob}.

\begin{remark}[Mixed-strategy uniqueness under strict curvature]\label{rmk:strict_mixed}
Suppose that $\lvert\ell_t''(z)\rvert<\ell_t'(z)(1-\ell_t'(z))$ for every $t\in\mathcal N$ and $z\in\mathbb R$. Then each manager's team-winning probability is strictly log-concave in her own interior allocation. For the power-form technology, this condition reduces to $r_t<1$ for every~$t$.

For team~$A$, interior battle probabilities imply $G_t>0$.\footnote{Using $\beta$ and $\gamma$ defined in \autoref{footnote}, $G_t=\frac{p_t(1-p_t)(\beta-\gamma)}{\beta p_t+\gamma(1-p_t)}>0$, because $\beta-\gamma=\Pr(\sum_{s\ne t}X_s=N)>0$ when all battle-winning probabilities are interior. The team-$B$ argument is identical after relabeling.} The strict upper curvature bound therefore gives
\[
    M-M^0
    =
    \operatorname{diag}\!\left(
      G_t\left[
      \frac{h_t'-h_t''}{(h_t')^2}-1
      \right]\right)
    \succ0.
\]
Since $M^0\succeq0$ by \autoref{prop:covdom_all}, \autoref{lemma:Hessian} implies that the Hessian of $\log\textnormal{Prob}_A$ is negative definite. The analogous argument for team~$B$ uses $G_t^B>0$ and $(h_t'+h_t'')/(h_t')^2>1$. Thus both managers' objectives are strictly log-concave on the interior.

Continuity up to the boundary then makes the equilibrium allocation each manager's unique best response to the opponent's equilibrium allocation: any distinct boundary optimum would, by concavity and continuity, make the intervening interior segment optimal, contradicting strict log-concavity. Finally, equilibrium interchangeability in this constant-sum game \citep{Klumpp-2006-Primaries} implies that any mixed equilibrium could be paired with the opponent's pure equilibrium strategy. Its support would therefore consist of best responses to that strategy; because the best response is unique, the mixed strategy must be degenerate. Hence the pure-strategy equilibrium is also the unique Nash equilibrium when mixed strategies are allowed.
\end{remark}

\section{Discussion}\label{sec:discussion}

In this section, we discuss the roles of \autoref{ass:battle_technology} in \autoref{subsec:robust}, conduct comparative statics in \autoref{Section:CS}, and draw out implications in \autoref{subsec:implication}. The comparative statics are stated for the power-form Tullock contest technology, under which the equilibrium objects admit explicit formulas.

\subsection{The Roles of \autoref{ass:battle_technology}}\label{subsec:robust}

The three parts of \autoref{ass:battle_technology} play distinct roles: homogeneity of degree zero shapes the equilibrium, the \textbf{curvature condition} secures its existence, and boundary regularity rules out corner allocations. Together, they establish \autoref{thm:main} for a broad class of homogeneous-of-degree-zero technologies satisfying the boundary and curvature conditions, with the power form as our running illustration. We revisit each in turn.

\textbf{Homogeneity of degree zero} delivers the equilibrium's \emph{structure}. The single-battle reduced form depends only on the stake ratio (\autoref{lem:csf}), and the formula \eqref{eq:Probi} built on \autoref{lem:local_reduction} assembles the full-contest winning probabilities from these reduced forms; hence only relative incentives matter within a battle. Combined with the common multiplier $\pi^{(t)}$, this drives both managers' first-order conditions into a common form, yielding the alignment property (\autoref{pro:PPP}) and the salience-share formula (\autoref{thm:main}). Conditional on the regularity conditions that secure interiority and global optimality, this structural argument uses no functional form beyond homogeneity of degree zero; the power form merely makes the pieces explicit. Homogeneity is a substantive restriction rather than a normalization: for technologies outside the HD-0 class---difference-form success functions, for example---the reduced-form winning probability depends on the level of stakes and not only on their ratio, the two managers' first-order conditions lose their common structure, and alignment should not be expected.

\textbf{Curvature condition}~\eqref{eq:curvature} delivers \emph{existence}. Beyond guaranteeing a unique interior equilibrium in each within-battle effort game, it ensures that each manager's team-winning probability is log-concave in her own allocation (\autoref{thm:contprob}), turning the first-order conditions into global best responses. The bound is two-sided, its halves playing mirror-image roles across teams: the upper bound $\ell_t''\le\ell_t'(1-\ell_t')$, equivalently $h_t'-h_t''\ge(h_t')^2$, governs team~$A$, and the lower bound $\ell_t''\ge-\ell_t'(1-\ell_t')$, equivalently $h_t'+h_t''\ge(h_t')^2$, governs team~$B$. Each half enters the log-concavity argument at a single point: the relaxation from $M\succeq 0$ to $M^0\succeq 0$ in \autoref{subsec:logconcavity}. The combinatorial core (\autoref{prop:covdom_all}) and the sign property $G_t\ge 0$ are technology-free.

Finally, \textbf{boundary regularity} delivers \emph{interiority}. Through the reduced-form extension $p_{At}(0,v_{Bt})=0$ for $v_{Bt}>0$ and $p_{At}(v_{At},0)=1$ for $v_{At}>0$, it makes the no-boundary property (\autoref{lem:no_boundary}) hold regardless of the specific form of battle technology. If one team assigns zero reward while its rival assigns a positive amount, the rival already wins that battle with certainty and can profitably redirect a sufficiently small amount elsewhere. If both assign zero, an arbitrarily small reward creates a discrete local advantage. These forces rule out corner allocations.

\subsection{Comparative Statics}\label{Section:CS}

We use the leading power-form example to examine how the equilibrium of \autoref{thm:main} responds to changes in primitives, organizing the analysis around two objects: the team-winning probability and the battle salience.

\paragraph{Team-winning probability.}
We now show that a team's equilibrium winning probability rises when its reward budget increases, or any member's marginal cost falls.

\begin{proposition}\label{strength}
Let $\textnormal{Prob}_A^{*}$ denote team~$A$'s equilibrium contest-winning probability.
\begin{enumerate}
    \item[(i)] For every $t\in\mathcal N$, $\textnormal{Prob}_A^{*}$ is strictly increasing in $c_t\triangleq\left(\frac{c_{Bt}}{c_{At}}\right)^{r_t}$. Equivalently, holding $c_{Bt}$, $r_t$, and all other battle primitives fixed, $\textnormal{Prob}_A^{*}$ is strictly decreasing in $c_{At}$.
    \item[(ii)] $\textnormal{Prob}_A^{*}$ is strictly decreasing in $k\triangleq\frac{W_B}{W_A}$. Equivalently, holding $W_B$ and all cost and technology primitives fixed, $\textnormal{Prob}_A^{*}$ is strictly increasing in $W_A$.
\end{enumerate}
\end{proposition}

The symmetric conclusions hold for team~$B$ after interchanging the team labels.

Part~(i) embodies a \emph{monotone strengthening property}: strengthening any single player always raises the team's winning probability, even though equilibrium prize allocations adjust across \emph{all} battles in response. This is not obvious \emph{a priori}, since improving player $A(t)$ could in principle absorb prize budget from other battles and thereby undermine the team's performance elsewhere. The alignment property (\autoref{pro:PPP}) closes this channel: because $v_{Bt}^{*}/v_{At}^{*} = k$ is independent of cost ratios, each $p_{At}^{*}$ depends only on the corresponding battle's fundamentals and the common budget ratio~$k$, so a unilateral reduction in $c_{At}$ feeds through to $\textnormal{Prob}_A^{*}$ without disturbing the winning odds in any other battle. Part~(ii) records the parallel and unambiguous benefit of a larger reward budget.

Under equal budgets, $W_A=W_B$, if $c_t=1$ for every battle, then $\textnormal{Prob}_A^*=1/2$. Still under equal budgets, if $c_t\geq1$ for every $t$, with at least one strict inequality, then $\textnormal{Prob}_A^*>1/2$. A natural conjecture is therefore that the product $\prod_t c_t$ of battle-level cost advantages summarizes which team is stronger overall. This conjecture is false. For instance, with $2N+1=3$, $r_t=1$ for every $t$, and $(c_1,c_2,c_3)=(99,\,99,\,1/9801)$, we have $\prod_t c_t=1$, yet $\textnormal{Prob}_A^*\approx0.98$.

\paragraph{Battle salience.} Recall that the salience of battle~$t$ decomposes as
\[
S_t \;=\; \theta^*(t)\cdot r_t \cdot \mathrm{Clo}_t ,
\]
where $\mathrm{Clo}_t \triangleq p_{At}^* p_{Bt}^*$ is battle $t$'s closeness and $\theta^*(t)$ is its pivotality. Recall that $\rho_t = c_{Bt}/c_{At}$ denotes the cost ratio of the matched pair, so that $c_t = \rho_t^{r_t}$, and that $k = W_B/W_A$ denotes the budget ratio of the two teams.

\begin{proposition}[Comparative Statics of Battle Salience and Closeness]
\label{prop:salience_cs}
For every $t\in\mathcal N$:
\begin{enumerate}
    \item[(i)] Holding $r_t$, $k$, and the primitives of all other battles fixed, $S_t$ is strictly single-peaked in $\rho_t$, attaining its unique maximum at $\rho_t=k$, where $p_{At}^*=p_{Bt}^*=1/2$. Moreover, $S_t\to0$ as $\rho_t\to0$ or $\rho_t\to\infty$.

    \item[(ii)] Holding $\rho_t=k$ and the primitives of all other battles fixed, $S_t$ is strictly increasing in $r_t$.

    \item[(iii)] Holding $\{c_s,r_s\}_{s\in\mathcal N}$ fixed,
    \[
        \frac{\partial\log\mathrm{Clo}_t}{\partial\log k}
        =
        r_t\bigl(p_{At}^*-p_{Bt}^*\bigr).
    \]
    Hence, the own-battle closeness rises with team~$B$'s relative budget if and only if team~$A$ is the more likely winner of battle~$t$ in equilibrium.
\end{enumerate}
\end{proposition}

Part~(i) formalizes the recurring intuition that ``equilibrium rewards favor intermediate battles'': salience is maximized at the cost ratio that equalizes equilibrium winning probabilities, not at the cost ratio that equalizes marginal costs. The peak location $\rho_t=k$ shifts with the budget ratio, so when team~$A$ enjoys a budget advantage ($k<1$), the salience-maximizing cost ratio satisfies $\rho_t<1$, meaning that team~$B$ enjoys a cost advantage.

Part~(ii) isolates the cleanest case in which discriminatory power affects salience. When the matched players are equalized at the prevailing budget ratio ($\rho_t=k$), the closeness component is fixed at $1/4$, and raising $r_t$ increases salience through the discriminatory-power component. Away from this benchmark, raising $r_t$ also widens any gap between $\rho_t$ and $k$, creating an opposing closeness effect that can offset the direct effect.

Part~(iii) concerns the own-battle closeness component rather than total salience. Raising team~$B$'s relative budget shifts the local winning probabilities toward team~$B$, increasing closeness when team~$A$ was initially favored and decreasing it otherwise. A change in $k$ also changes $\theta^*(t)$ through the equilibrium probabilities in the other battles, so the response of total salience $S_t$ to $k$ is generally ambiguous.

\subsection{Implications and Applications}\label{subsec:implication}

The equilibrium formula yields a direct allocation principle: a battle commands a large reward share when its discriminatory power, closeness, and pivotality are simultaneously high. Closeness alone is insufficient if the battle is unlikely to affect the team outcome, and pivotality alone is insufficient if the battle is too lopsided for marginal incentives to matter.

The alignment property yields the paper's sharper organizational implication. Rival managers facing the same battles and majority objective choose the same normalized reward schedule. Budgets, player costs, and technological biases affect that common schedule through equilibrium probabilities and salience, but they do not generate manager-specific relative tilts. The analysis also produces two comparative-static implications: strengthening a player raises the team's overall winning probability, and timing is irrelevant for allocations and ex ante winning probabilities when rewards are precommitted.

\paragraph{Procurement and consortium competition.}
Procurement provides the closest institutional match. In a typical defense procurement, rival prime contractors compete for a single award, while performance is assessed across specialized dimensions such as propulsion, avionics, software, integration, reliability, and cost. Contractually promised workshare or revenue shares can be made contingent on winning the procurement and therefore map naturally into $v_{it}$. The model predicts that larger shares go to subsystems in which performance is responsive to effort, the rival teams are closely matched, and the subsystem is likely to determine the award. Subsystems that are nearly predetermined or unlikely to affect the final evaluation receive smaller incentive shares.

The majority rule is an approximation because many procurements use weighted scores, technical thresholds, and holistic assessments, and subsystem performances may interact. The application is nevertheless useful because it preserves the central organizational mechanism: a collective award is produced through decentralized specialized effort, and the internal division of the award is strategically chosen against the rival.

\paragraph{Team sports.}
A federation, club, or team captain may allocate a fixed championship bonus pool across players assigned to distinct matches before competition begins. Such bonuses are naturally victory-contingent, and a series such as the Davis Cup or Ryder Cup resembles a collection of pairwise battles aggregated into team victory. The cleanest mapping has match assignments fixed in advance, with each rewarded player or pair associated with one battle; repeated participation across matches would require an extension of the model. The salience rule predicts larger promised shares for matches that are both responsive to effort and likely to decide the series. Where match-specific bonus schedules are observed, alignment offers a particularly direct empirical implication: after normalization by pool size, the two schedules should coincide match by match. The timing result applies only when the bonus schedule is fixed in advance; renegotiating rewards after early results would create a different dynamic game.

\paragraph{A looser political interpretation.}
Legislative elections share the majority structure, but the closest counterpart to $v_{it}$ is not campaign spending. It is a credible claim on the rents from majority control---for example, committee leadership, agenda access, policy concessions, patronage, or future party support. Such promises need not be legally binding; they may be sustained through repeated interaction, party discipline, and reputational concerns, as in a relational contract. Under this interpretation, the model suggests that both parties' promised post-victory rewards should place similar normalized weight on races that are contestable and pivotal for chamber control. This is the reward-contract counterpart of swing-state targeting in \citet*{stromberg2008electoral}: both sides concentrate on races that are close and pivotal, but the targeting here operates through promised claims on majority rents rather than through campaign resources.

Preelection transfers, independent expenditures, and leadership visits are inputs supplied before the collective outcome rather than rewards paid only after victory. They may also affect effort costs directly. Fundraising and campaign spending are therefore better viewed as measures of candidate strength or race closeness, or as objects for an extended model with direct effort subsidies, rather than as literal proxies for $v_{it}$; empirical work on vulnerability and fundraising illustrates these distinct channels \citep*{dynes2015partisanship,thomsen2023competition}. Two further gaps between the model and this application are worth naming. A winning candidate retains her seat regardless of chamber control, so the seat itself is an automatic battle-level prize outside the party's schedule---a force absent from procurement and bonus-pool settings, where winning one's own dimension pays nothing without the collective award. Any role for minority-state promises must also be interpreted subject to the contractual qualifications discussed in \autoref{Section:Model}. The election application should consequently be read as a qualitative relational-contract interpretation, not a literal campaign-finance model.

\section{Concluding Remarks}
\label{Section:Conclusion}

This paper studies how rival organizations divide fixed, victory-contingent reward budgets across members who fight separate battles. In the unique pure-strategy equilibrium, the managers choose the same normalized reward schedule despite differences in team budgets, player costs, and battle technologies. We call this result \emph{reward schedule alignment}. The common schedule directs larger shares to battles that combine greater discriminatory power, closer competition, and higher pivotality for team victory. Strategic competition therefore produces mirroring rather than specialization.

The main technical challenge is global equilibrium existence. Under the regularity condition on battle technologies, each manager's team-winning probability is log-concave in her own reward allocation, making the first-order characterization globally valid for any odd number of battles. This yields a unique pure-strategy Nash equilibrium. When the curvature condition holds strictly, strict log-concavity also rules out nontrivial mixing, so the same profile is the unique Nash equilibrium even when mixed strategies are allowed.

Three extensions would clarify the reach of reward schedule alignment. First, weighted aggregation and other victory rules would separate majority pivotality from the broader incentive mechanism. Second, richer contracts, including outcome-contingent rewards and history-dependent renegotiation, would introduce within-team tournament design and a meaningful role for commitment. Third, one could relax homogeneity of degree zero, which is also the technological foundation of \citet*{Fu-2015-Team}. By eliminating level effects, HD-0 underpins both their temporal-neutrality results and the alignment obtained here. Outside this class, absolute reward levels affect local outcomes, the managers' optimality conditions need not share a common structure, and alignment may fail.


\newpage
\appendix
\begin{center}
{\LARGE Appendix}
\par\end{center}

\section{Proofs for \autoref{Section:Model}}

This part collects the proofs of \autoref{lem:csf}, \autoref{lem:local_reduction}, \autoref{lem:temporal}, and \autoref{lem:no_boundary}.

\subsection*{Proof of \autoref{lem:csf}}

Fix battle~$t$ and positive stakes $(x_A,x_B)$. Write $p\equiv\widetilde p_{At}$ and $z\equiv\log(e_A/e_B)$. The log-odds representation gives
\[
\frac{\partial p}{\partial e_A}
=\frac{p(1-p)\ell_t'(z)}{e_A},
\qquad
\frac{\partial(1-p)}{\partial e_B}
=\frac{p(1-p)\ell_t'(z)}{e_B}.
\]
The first-order conditions set these derivatives, multiplied by $x_A$ and $x_B$, equal to $c_{At}$ and $c_{Bt}$. Dividing them yields
\[
\frac{e_A}{e_B}
=\frac{x_Ac_{Bt}}{x_Bc_{At}},
\qquad
z^*
=\log\frac{x_A}{x_B}+\log\frac{c_{Bt}}{c_{At}}.
\]
Let $p^*\equiv[1+\exp(-\ell_t(z^*))]^{-1}$ and define
\[
e_A^*=\frac{x_Ap^*(1-p^*)\ell_t'(z^*)}{c_{At}},
\qquad
e_B^*=\frac{x_Bp^*(1-p^*)\ell_t'(z^*)}{c_{Bt}}.
\]
These efforts are positive, have log ratio $z^*$, and satisfy both first-order conditions. The curvature bounds make each player's objective strictly concave in her own positive effort, so the conditions are sufficient. Since $0<\ell_t'\leq1$, the resulting payoffs are
\[
x_Ap^*\bigl[1-(1-p^*)\ell_t'(z^*)\bigr]>0,
\qquad
x_B(1-p^*)\bigl[1-p^*\ell_t'(z^*)\bigr]>0,
\]
whereas zero effort against a positive effort yields zero. Thus $(e_A^*,e_B^*)$ is an interior equilibrium.

No equilibrium can involve zero effort. If both efforts were zero, either player could profitably choose a sufficiently small positive effort; if exactly one were positive, that player could lower it while preserving certain victory and hence would have no optimal positive effort. Every equilibrium is therefore interior, and the first-order conditions uniquely determine first $z^*$ and then $(e_A^*,e_B^*)$.

Finally,
\[
\log\frac{p_{At}(x_A,x_B)}{p_{Bt}(x_A,x_B)}
=\ell_t\!\left(\log\frac{x_A}{x_B}+\log\frac{c_{Bt}}{c_{At}}\right)
\equiv h_t\!\left(\log\frac{x_A}{x_B}\right).
\]
Because $\ell_t$ is strictly increasing and twice continuously differentiable, so is $h_t$. The expression depends on $(x_A,x_B)$ only through $x_A/x_B$, proving homogeneity of degree zero and strict monotonicity in the stake ratio.

\subsection*{Proof of \autoref{lem:local_reduction}}

Fix a precommitted allocation profile and a history at which battle~$t$ is reached.

\emph{Proof of (i).} Holding all other future strategies and battle probabilities fixed, define
\[
\pi_A^{(t)}
\equiv
\Pr(A\text{ wins contest}\mid A\text{ wins battle }t)
-
\Pr(A\text{ wins contest}\mid A\text{ loses battle }t),
\]
and define $\pi_B^{(t)}$ analogously. Because team victories are complementary and $B$ wins battle~$t$ exactly when $A$ loses it,
\[
\pi_B^{(t)}
=\bigl[1-\Pr(A\text{ wins contest}\mid A\text{ loses battle }t)\bigr]
 -\bigl[1-\Pr(A\text{ wins contest}\mid A\text{ wins battle }t)\bigr]=\pi_A^{(t)}.
\]
Let $\pi^{(t)}$ denote this common value. It equals the conditional probability that the outcomes of the other battles leave battle~$t$ decisive. Hence $\pi^{(t)}>0$ whenever battle~$t$ can still affect the identity of the winning team.

\emph{Proof of (ii).} Up to terms independent of their own efforts, the two players' expected payoffs are
\[
\begin{aligned}
&\pi^{(t)}v_{At}\,\widetilde p_{At}(e_{At},e_{Bt})-c_{At}e_{At},\\
&\pi^{(t)}v_{Bt}\,\bigl(1-\widetilde p_{At}(e_{At},e_{Bt})\bigr)-c_{Bt}e_{Bt},
\end{aligned}
\]
respectively. Thus the battle-$t$ effort subgame coincides with the auxiliary one-battle contest with stakes $(\pi^{(t)}v_{At},\pi^{(t)}v_{Bt})$. By homogeneity of degree zero in \autoref{lem:csf},
\[
p_{At}\bigl(\pi^{(t)}v_{At},\pi^{(t)}v_{Bt}\bigr)
=
p_{At}(v_{At},v_{Bt}).
\]
At one-sided allocations, the battle outcome is defined by the reduced-form completion specified above rather than by an equilibrium of the corresponding effort subgame.

\subsection*{Proof of \autoref{lem:temporal}}

Suppose that the temporal structure contains $Z\le 2N+1$ clusters. Let $\mathcal N^z$ denote the set of battles in the first $z$ clusters, and let $i^z$ denote the event that team~$i$ has secured at least $N+1$ victories by the end of cluster~$z$. If the contest has stopped before a cluster is reached, complete that cluster by independent artificial draws with probabilities $(p_{At},p_{Bt})$. These draws do not affect the identity of a team that has already secured a majority.

We prove by induction on $z$ that
\[
\Pr(i^z)
=
\sum_{\substack{w_i\subseteq\mathcal N^z\\|w_i|\ge N+1}}
\prod_{t\in w_i}p_{it}
\prod_{q\in\mathcal N^z\setminus w_i}p_{jq},
\qquad j\in\{A,B\}\setminus\{i\}.
\]

\noindent\emph{Base case.}
For $z=1$, every battle in the first cluster is played simultaneously. Conditional independence and \autoref{lem:local_reduction} imply that its battle outcomes have product probabilities $(p_{At},p_{Bt})$, so the formula follows by direct enumeration. If the first cluster contains fewer than $N+1$ battles, both sides are zero.

\medskip
\noindent\emph{Inductive step.}
Suppose the formula holds after cluster~$z$. Conditional on any history at which neither team has yet secured a majority, \autoref{lem:local_reduction} and conditional independence imply that the outcomes in the next cluster have product probabilities generated by $(p_{At},p_{Bt})$. Team~$i$ has secured a majority after cluster~$z+1$ in exactly two mutually exclusive ways: it already had a majority after cluster~$z$, or it reaches a majority for the first time in cluster~$z+1$. Hence
\[
\Pr(i^{z+1})
=
\Pr(i^z)
+
\sum_{\substack{w_i\subseteq\mathcal N^{z+1}\\
|w_i|\ge N+1\\
|w_i\cap\mathcal N^z|\le N}}
\prod_{t\in w_i}p_{it}
\prod_{q\in\mathcal N^{z+1}\setminus w_i}p_{jq}.
\]
For histories in $i^z$, complete the battles in cluster~$z+1$ by the independent artificial draws described above. Substituting the inductive hypothesis for $\Pr(i^z)$ and expanding each of its terms over all outcomes in the new cluster, using $p_{it}+p_{jt}=1$, yields the product-probability sum over coalitions satisfying $|w_i\cap\mathcal N^z|\ge N+1$. The second sum covers the complementary winning coalitions satisfying $|w_i\cap\mathcal N^z|\le N$. Together they exhaust all $w_i\subseteq\mathcal N^{z+1}$ with $|w_i|\ge N+1$, proving the induction step.

At $z=Z$, $\mathcal N^Z=\mathcal N$, so the induction formula becomes~\eqref{eq:Probi}. Therefore, the ex ante team-winning probability is independent of the temporal structure.

\subsection*{Proof of \autoref{lem:no_boundary}}

Fix a candidate pure-strategy equilibrium and classify battles by
\[
D_A=\{t:p_{At}=1\},\qquad D_B=\{t:p_{At}=0\},\qquad R=\{t:p_{At}\in(0,1)\},
\]
so $D_A$ ($D_B$) consists of battles funded only by team~$A$ ($B$), while $R$ contains contested and unfunded battles, the latter having $p_{At}=1/2$. Call battle~$t$ \emph{pivotal} if the remaining $2N$ battles split evenly.

\emph{Interior contest probabilities.} If $\textnormal{Prob}_A=0$, team~$A$ can deviate to the full-support allocation $\tilde v_{At}=W_A/(2N+1)$, giving it a positive probability of winning every battle and hence the contest, a contradiction. Symmetrically, $\textnormal{Prob}_B>0$, so $\textnormal{Prob}_A,\textnormal{Prob}_B\in(0,1)$.

\emph{Positive pivotality.} Let $a\triangleq|D_A|$ and $r\triangleq|R|$. Because both teams have positive contest-winning probabilities, team~$A$ has at most $N$ sure wins and at least $N+1$ battles it can win with positive probability. Thus
\[
a\leq N
\qquad\text{and}\qquad
a+r\geq N+1.
\]
If $t\in D_A$, excluding $t$ leaves $a-1$ sure wins, and the remaining battles split evenly when $A$ wins exactly $N-a+1$ of the $r$ battles in $R$; this is feasible because $1\leq N-a+1\leq r$. If $t\in D_B$, the required count is $N-a$ of the $r$ battles in $R$, with $0\leq N-a\leq r$. Finally, if $t\in R$, the required count is $N-a$ of the other $r-1$ battles in $R$, with $0\leq N-a\leq r-1$. All inequalities follow from the two bounds above. Since every battle in $R$ is won by either team with strictly positive probability, independence makes each feasible exact-count event strictly positive. Hence $\theta(t)>0$ for every $t\in\mathcal N$.

\emph{No one-sided allocation.} Suppose $v_{At}=0<v_{Bt}$ for some $t$. Since $\textnormal{Prob}_B<1$, some battle $s$ is not a sure win for $B$. Shifting a small $\varepsilon>0$ from $v_{Bt}$ to $v_{Bs}$ keeps $t$ a sure win while strictly raising $p_{Bs}$. Because $s$ is pivotal with positive probability, this raises $\textnormal{Prob}_B$, a contradiction.

\emph{No unfunded battle.} Suppose $v_{At}=v_{Bt}=0$. Since one-sided allocations have been excluded and budgets are positive, some battle $s$ has $v_{As},v_{Bs}>0$. Team~$A$ shifts $\varepsilon>0$ from $s$ to $t$. By continuity at the interior allocation in battle~$s$, its probability loss vanishes as $\varepsilon\downarrow0$, whereas $t$ jumps from a coin flip to a sure win, yielding a limiting gain of $\theta(t)/2>0$. Thus a sufficiently small $\varepsilon$ is profitable, a contradiction.

Therefore, every pure-strategy equilibrium is interior.

\section{Proofs for \autoref{Section:Equi}}

This part collects the proofs of \autoref{lemma:Hessian} and \autoref{prop:covdom_all}.

\subsection*{Proof of \autoref{lemma:Hessian}}

Define the log-odds $\eta_t = \log\bigl(p_t/(1-p_t)\bigr)$, so that $\mathbf{X}$ has the canonical exponential-family law
\[
\Pr(\mathbf{x} \mid \boldsymbol{\eta}) = \prod_{t=1}^{2N+1} \frac{e^{\eta_t x_t}}{1+e^{\eta_t}}.
\]

\noindent\emph{First-order derivative.} Writing $f = \sum_{\mathbf{x} \in A} \Pr(\mathbf{x})$ and using $\partial_{\eta_t} \log[e^{\eta_t x_t}/(1+e^{\eta_t})] = x_t - p_t$,
\begin{equation*}
\frac{\partial \log f}{\partial \eta_t} = \mathbb{E}[X_t \mid A] - p_t \eqqcolon G_t.
\end{equation*}

\medskip

\noindent\emph{Second-order derivative.} By the same identity, $\partial_{\eta_s} \mathbb{E}[X_t \mathbf{1}_A] = \mathbb{E}[X_t X_s \mathbf{1}_A] - p_s\, \mathbb{E}[X_t \mathbf{1}_A]$. Applying the quotient rule to $\mathbb{E}[X_t \mid A] = \mathbb{E}[X_t \mathbf{1}_A]/f$ and using $\partial_{\eta_s} f = f \cdot G_s$, the $p_s \mathbb{E}[X_t \mid A]$ terms cancel and the result simplifies to $\operatorname{Cov}(X_t, X_s \mid A)$. Since $\partial p_t / \partial \eta_s = \delta_{ts} p_t(1-p_t)$,
\begin{equation*}
\frac{\partial^2 \log f}{\partial \eta_t\, \partial \eta_s} = \operatorname{Cov}(X_t, X_s \mid A) - \delta_{ts}\, p_t(1-p_t).
\end{equation*}

\medskip

\noindent\emph{Chain rule to prize shares.} By HD-0, $\eta_t = h_t(u_t)$ with $u_t = \log v_{At}-\log v_{Bt}$, where the dependence on $v_{Bt}$ is held fixed; hence $\eta_t$ depends on $\mathbf{v}_A$ only through $v_{At}$, with
\begin{equation*}
\frac{\partial \eta_t}{\partial v_{At}} = \frac{h_t'}{v_{At}}, \qquad \frac{\partial^2 \eta_t}{\partial v_{At}^2} = \frac{h_t''-h_t'}{v_{At}^2}.
\end{equation*}
The chain rule gives $\partial_{v_{At}} \log f = G_t h_t' / v_{At}$. Differentiating again with respect to $v_{As}$, and noting that the second-derivative term $\partial \log f/\partial \eta_t \cdot \partial^2 \eta_t/\partial v_{As}^2$ contributes only when $t = s$,
\begin{equation}\label{eq:HvApp}
H_{ts} = \frac{h_t' h_s'}{v_{At} v_{As}} \bigl[\operatorname{Cov}(X_t, X_s \mid A) - \delta_{ts}\, p_t(1-p_t)\bigr] + \delta_{ts}\, \frac{(h_t''-h_t')\,G_t}{v_{At}^2}.
\end{equation}
Setting $A_d = \operatorname{diag}(h_t'/v_{At})$, $D = \operatorname{diag}(p_t(1-p_t))$, and $\Sigma_A = \operatorname{Cov}(\mathbf{X} \mid A)$, equation \eqref{eq:HvApp} reads
\begin{equation}\label{eq:HvGen}
H = -A_d(D-\Sigma_A)A_d - \operatorname{diag}\!\Bigl(G_t\,\tfrac{h_t'-h_t''}{v_{At}^2}\Bigr).
\end{equation}
Under the invertible change of variables $\widetilde{\mathbf{u}} \triangleq A_d \mathbf{u}$,
\[
\mathbf{u}^\top H \mathbf{u} = -\widetilde{\mathbf{u}}^\top (D-\Sigma_A) \widetilde{\mathbf{u}} - \sum_t \widetilde{u}_t^2\, G_t\,\frac{h_t'-h_t''}{(h_t')^2} = -\widetilde{\mathbf{u}}^\top M \widetilde{\mathbf{u}}.
\]

\subsection*{Proof of \autoref{prop:covdom_all}}


\subsubsection*{Step 1: Loss reformulation, entries of $M^0$, and reduction to $\widetilde M$}

\paragraph{Loss reformulation.}
Let $Y_t = 1 - X_t$ and $\phi_t = (1-p_t)/p_t > 0$, so the majority event becomes $A = \{\sum_t Y_t \leq N\}$. For $U \subseteq \mathcal{N}$, write $\boldsymbol{\phi}_{-U}$ for the subvector obtained by removing the coordinates in $U$, abbreviating $\boldsymbol{\phi}_{-\{t\}}$ as $\boldsymbol{\phi}_{-t}$. Let $e_k(\boldsymbol{\phi}) = \sum_{|S|=k} \prod_{t \in S} \phi_t$ denote the elementary symmetric polynomial of degree $k$, with $e_0 = 1$ and $e_k = 0$ outside the natural range, and set
\[
\Phi_m(\boldsymbol{\phi}) = \sum_{k=0}^{m} e_k(\boldsymbol{\phi}), \qquad \Phi_m = 0 \text{ for } m < 0.
\]
Independence of the $Y_t$ gives $\Pr(\sum_t Y_t = k) = (\prod_t p_t)\, e_k(\boldsymbol{\phi})$, so $\Pr(A) = (\prod_t p_t)\, \Phi_N(\boldsymbol{\phi})$, and for any $|S| \leq N$,
\[
\Pr\bigl(\{t:Y_t=1\}=S\mid A\bigr) = \frac{\prod_{t \in S} \phi_t}{\Phi_N(\boldsymbol{\phi})}.
\]

\paragraph{Closed-form entries of $M^0$.}
Let $a_t \triangleq \Pr(Y_t = 1 \mid A)$, $b_t \triangleq \phi_t/(1+\phi_t)$, $q_t \triangleq 1 - a_t$, and $p_t \triangleq 1 - b_t$. Since $M^0 = D + \operatorname{diag}(G_t) - \Sigma_A$ with $G_t = q_t - p_t$ and $\operatorname{Var}(X_t \mid A) = q_t(1-q_t)$, the diagonal simplifies to
\[
M^0_{tt} = p_t(1-p_t) + (q_t - p_t) - q_t(1-q_t) = q_t^2 - p_t^2,
\]
while $M^0_{ts} = -\operatorname{Cov}(Y_t, Y_s \mid A)$ for $t \neq s$. Summing over winning configurations containing $t$,
\[
a_t = \sum_{S \ni t,\, |S| \leq N} \frac{\prod_{s \in S}\phi_s}{\Phi_N(\boldsymbol{\phi})} = \frac{\phi_t\, \Phi_{N-1}(\boldsymbol{\phi}_{-t})}{\Phi_N(\boldsymbol{\phi})}.
\]
Conditioning on whether battle $t$ is lost yields the splitting identity
\begin{equation}\label{eq:split}
\Phi_N(\boldsymbol{\phi}) = \Phi_N(\boldsymbol{\phi}_{-t}) + \phi_t\, \Phi_{N-1}(\boldsymbol{\phi}_{-t}),
\end{equation}
hence $q_t = \Phi_N(\boldsymbol{\phi}_{-t})/\Phi_N(\boldsymbol{\phi})$. Using \eqref{eq:split},
\[
b_t - a_t = \frac{\phi_t}{1+\phi_t} - \frac{\phi_t \Phi_{N-1}(\boldsymbol{\phi}_{-t})}{\Phi_N(\boldsymbol{\phi})} = \frac{\phi_t\, e_N(\boldsymbol{\phi}_{-t})}{(1+\phi_t)\Phi_N(\boldsymbol{\phi})},
\]
where the second equality follows from $\Phi_N(\boldsymbol{\phi}) - (1+\phi_t)\Phi_{N-1}(\boldsymbol{\phi}_{-t}) = \Phi_N(\boldsymbol{\phi}_{-t}) - \Phi_{N-1}(\boldsymbol{\phi}_{-t}) = e_N(\boldsymbol{\phi}_{-t})$. A parallel calculation gives
\[
q_t + p_t = \frac{2\Phi_N(\boldsymbol{\phi}) + \phi_t\, e_N(\boldsymbol{\phi}_{-t})}{(1+\phi_t)\Phi_N(\boldsymbol{\phi})},
\]
and multiplying via $M^0_{tt} = (q_t - p_t)(q_t + p_t)$,
\begin{equation}\label{eq:general_diag_M0}
M^0_{tt} = \frac{\phi_t\, e_N(\boldsymbol{\phi}_{-t})\bigl(2\Phi_N(\boldsymbol{\phi}) + \phi_t\, e_N(\boldsymbol{\phi}_{-t})\bigr)}{(1+\phi_t)^2 \Phi_N(\boldsymbol{\phi})^2}.
\end{equation}
For $t \neq s$, $M^0_{ts} = a_t a_s - \mathbb{E}[Y_t Y_s \mid A]$, with
\[
\mathbb{E}[Y_t Y_s \mid A] = \frac{\phi_t \phi_s\, \Phi_{N-2}(\boldsymbol{\phi}_{-\{t,s\}})}{\Phi_N(\boldsymbol{\phi})}.
\]
Combining this with $a_t a_s = \phi_t \phi_s\, \Phi_{N-1}(\boldsymbol{\phi}_{-t}) \Phi_{N-1}(\boldsymbol{\phi}_{-s}) / \Phi_N(\boldsymbol{\phi})^2$ and applying \eqref{eq:split} to expand $\Phi_{N-1}(\boldsymbol{\phi}_{-t})$ and $\Phi_N(\boldsymbol{\phi})$ relative to index $s$, the $\phi_t$- and $\phi_s$-cross-terms cancel and the bracket reduces to the \emph{Tur\'an defect}
\[
T_N(\boldsymbol{\psi}) \triangleq \Phi_{N-1}(\boldsymbol{\psi})^2 - \Phi_{N-2}(\boldsymbol{\psi}) \Phi_N(\boldsymbol{\psi}),
\]
evaluated at $\boldsymbol{\psi} = \boldsymbol{\phi}_{-\{t,s\}}$. Hence
\begin{equation}\label{eq:general_offdiag_M0}
M^0_{ts} = \frac{\phi_t \phi_s\, T_N(\boldsymbol{\phi}_{-\{t,s\}})}{\Phi_N(\boldsymbol{\phi})^2} \qquad (t \neq s).
\end{equation}

\paragraph{Reduction to a polynomial matrix $\widetilde M$.}
Let $D_\phi \triangleq \operatorname{diag}(1+\phi_t)$ and $\widetilde M \triangleq \Phi_N(\boldsymbol{\phi})^2\, D_\phi M^0 D_\phi$. Since $\Phi_N > 0$ and $D_\phi \succ 0$ on the positive orthant, this congruence preserves inertia, so $M^0 \succeq 0 \iff \widetilde M \succeq 0$. Clearing denominators in \eqref{eq:general_diag_M0}--\eqref{eq:general_offdiag_M0},
\begin{align}
\widetilde M_{tt} &= \phi_t\, e_N(\boldsymbol{\phi}_{-t}) \bigl(2\Phi_N(\boldsymbol{\phi}) + \phi_t\, e_N(\boldsymbol{\phi}_{-t})\bigr), \label{eq:Mtilde_entries_diag}\\
\widetilde M_{ts} &= \phi_t \phi_s (1+\phi_t)(1+\phi_s)\, T_N(\boldsymbol{\phi}_{-\{t,s\}}), \quad t \neq s. \label{eq:Mtilde_entries_offdiag}
\end{align}
Every entry of $\widetilde M$ is a polynomial in $\boldsymbol{\phi}$, which opens the door to a coefficientwise PSD argument.

\subsubsection*{Step 2: Block structure and closed forms for the coefficients of $\widetilde M$}

Let $n\triangleq 2N+1$. Expanding $\widetilde M$ in the monomial basis,
\[
\widetilde M(\boldsymbol{\phi}) = \sum_{\alpha \in \mathbb{Z}_{\geq 0}^n} C_\alpha\, \boldsymbol{\phi}^{\alpha}, \qquad \boldsymbol{\phi}^{\alpha} \triangleq \prod_t \phi_t^{\alpha_t},
\]
it suffices to show $C_\alpha \succeq 0$ for every $\alpha$, since then $v^\top \widetilde M v = \sum_\alpha \boldsymbol{\phi}^\alpha (v^\top C_\alpha v) \geq 0$ on the positive orthant. We adopt the conventions $\binom{m}{r} = 0$ for $r \notin \{0,\dots,m\}$ and $[x]_+ = \max\{x,0\}$.

\paragraph{Support and block structure of $C_\alpha$.}
By \eqref{eq:Mtilde_entries_diag}--\eqref{eq:Mtilde_entries_offdiag}, each variable $\phi_t$ appears in $\widetilde M$ to power at most $2$; hence only $\alpha \in \{0,1,2\}^n$ contribute. Partition the support of such $\alpha$ as
\[
S_1 = \{t : \alpha_t = 1\}, \quad S_2 = \{t : \alpha_t = 2\}, \quad b = |S_1|, \quad a = |S_2|, \quad L = N - a.
\]
If $\alpha_t = 0$, every monomial in row $t$ of $\widetilde M$ contains $\phi_t$ to at least the first power (via the leading $\phi_t$ factor), so it cannot contribute to $\boldsymbol{\phi}^\alpha$; the same holds column-wise. Therefore $C_\alpha$ is supported on $S_1 \cup S_2$.

The symmetry of $\Phi_N, e_N, T_N$ in their arguments, together with the fact that $\boldsymbol{\phi}^\alpha$ depends only on $S_1, S_2$ as sets, forces $C_\alpha$ to be invariant under permutations preserving $S_1$ and $S_2$ setwise. Hence its entries take only five values, classified by index-pair type:
\[
\begin{array}{c|cc}
 & t = s & t \neq s \\ \hline
t, s \in S_1 & d_1 & c_{11} \\
t, s \in S_2 & d_2 & c_{22} \\
t \in S_1,\, s \in S_2 & - & c_{12}
\end{array}
\]
Ordering $S_1$ before $S_2$,
\begin{equation}\label{eq:Calpha_block}
C_\alpha = \begin{pmatrix} d_1 I_b + c_{11}(J_b - I_b) & c_{12}\, \mathbf{1}_b \mathbf{1}_a^\top \\ c_{12}\, \mathbf{1}_a \mathbf{1}_b^\top & d_2 I_a + c_{22}(J_a - I_a) \end{pmatrix},
\end{equation}
where $J_m = \mathbf{1}_m \mathbf{1}_m^\top$.

\begin{lemma}[Coefficient counting]\label{lem:coefficient_counting}
Let $S_1$ and $S_2$ be disjoint index sets with $|S_1|=b$ and $|S_2|=a$, and define
\[
    m_{S_1,S_2}(\boldsymbol{\psi})
    \triangleq
    \prod_{t\in S_1}\psi_t
    \prod_{t\in S_2}\psi_t^2.
\]
Then
\[
    [m_{S_1,S_2}]\,
    e_N(\boldsymbol{\psi})\Phi_K(\boldsymbol{\psi})
    =
    \binom{b}{N-a}
    \mathbf 1_{\{\,2a+b-N\le K\,\}},
\]
where $[m]F$ denotes the coefficient of the monomial $m$ in the polynomial $F$, and $\binom{b}{N-a}=0$ when $N-a\notin\{0,\ldots,b\}$.
\end{lemma}

\begin{proof}
A contributing term is a pair of index sets $(U_1,U_2)$ such that $|U_1|=N$, $|U_2|\le K$, and
\[
    \left(\prod_{t\in U_1}\psi_t\right)
    \left(\prod_{t\in U_2}\psi_t\right)
    =
    m_{S_1,S_2}(\boldsymbol{\psi}).
\]
Every index in $S_2$ must belong to both $U_1$ and $U_2$, while every index in $S_1$ must belong to exactly one of them. Since $U_1$ already contains the $a$ indices in $S_2$, it must contain exactly $N-a$ indices from $S_1$. There are $\binom{b}{N-a}$ ways to choose them. The remaining indices in $S_1$ belong to $U_2$, so
\[
    |U_2|
    =
    a+\bigl[b-(N-a)\bigr]
    =
    2a+b-N.
\]
The chosen pair contributes if and only if $|U_2|\le K$, which gives the stated indicator.
\end{proof}

\paragraph{Closed forms for $d_1, d_2, c_{11}, c_{12}, c_{22}$.}
We claim
\begin{align}
d_1 &= 2 \binom{b-1}{L} \mathbf{1}_{\{b-1 \leq 2L\}}, \label{eq:d1}\\
d_2 &= 2 \binom{b}{L+1} \mathbf{1}_{\{b \leq 2L+1\}} + \binom{b}{L+1} \mathbf{1}_{\{b = 2L+2\}}, \label{eq:d2}\\
c_{ij} &= \Bigl[\binom{b-2+(i+j-2)}{L-1+(i+j-2)} - \binom{b-2+(i+j-2)}{L+(i+j-2)}\Bigr]_+, \quad (i,j) \in \{(1,1),(1,2),(2,2)\}; \label{eq:cij}
\end{align}
that is, $c_{11} = [\binom{b-2}{L-1} - \binom{b-2}{L}]_+$, $c_{12} = [\binom{b-1}{L} - \binom{b-1}{L+1}]_+$, and $c_{22} = [\binom{b}{L+1} - \binom{b}{L+2}]_+$. The diagonal formulas follow from \autoref{lem:coefficient_counting}, while the off-diagonal formulas follow from the Tur\'an-coefficient calculation below.

\medskip
\emph{(i) Off-diagonal coefficients.} For $t \neq s$, expanding $\phi_t\phi_s(1+\phi_t)(1+\phi_s) = \phi_t\phi_s + \phi_t^2\phi_s + \phi_t\phi_s^2 + \phi_t^2\phi_s^2$, each exponent pair $(\alpha_t, \alpha_s) \in \{1,2\}^2$ is selected from this prefactor with multiplicity $1$, so the coefficient of $\boldsymbol{\phi}^\alpha$ in $\widetilde M_{ts}$ equals the coefficient of the restricted monomial $\prod_{q\notin\{t,s\}}\phi_q^{\alpha_q}$ in $T_N(\boldsymbol{\phi}_{-\{t,s\}})$.

Fix a monomial in $\boldsymbol{\phi}_{-\{t,s\}}$ with $a'$ squared and $r$ first-power variables; set $L' = N - a'$. We claim its coefficient in $T_N$ equals
\begin{equation}\label{eq:Turan_coeff}
\bigl[\binom{r}{L'-1} - \binom{r}{L'}\bigr]_+.
\end{equation}
Indeed, expand $T_N=\Phi_{N-1}^2-\Phi_{N-2}\Phi_N$. In each $\Phi$-factor, the $a'$ squared variables must occupy slots in the index set because each $\Phi$ is squarefree, leaving residual budgets $L'-1,L'-1,L'-2,L'$. The calculation therefore reduces to comparing the admissible ways of assigning the $r$ first-power variables between the two factors; the positive part in \eqref{eq:Turan_coeff} arises when the first product admits an additional endpoint assignment. Writing $h$ for the number assigned to the first factor, the coefficient of the target monomial is
\[
\underbrace{\sum_{h: h \leq L'-1,\, r-h \leq L'-1} \binom{r}{h}}_{\Phi_{N-1}^2} - \underbrace{\sum_{h: h \leq L'-2,\, r-h \leq L'} \binom{r}{h}}_{\Phi_{N-2}\Phi_N}.
\]
Consider first the regime $L' \leq r$. The two index ranges are then $[A, B]$ and $[A-1, B-1]$ with $A = r-L'+1 \geq 1$ and $B = \min\{r, L'-1\} = L'-1$. If $r \leq 2L'-2$, both ranges are nonempty and overlap on $[A, B-1]$, where the binomials cancel, leaving $\binom{r}{B} - \binom{r}{A-1} = \binom{r}{L'-1} - \binom{r}{L'}$ (using $\binom{r}{r-L'} = \binom{r}{L'}$); if $r \geq 2L'-1$, both ranges are empty and the coefficient is $0$. Next, when $L' = r+1$, the first sum runs over all of $[0,r]$ and equals $2^r$, while the second omits only the term $h = r$ and equals $2^r - 1$, so the coefficient is $1 = \binom{r}{L'-1} - \binom{r}{L'}$. Finally, when $L' \geq r+2$, both sums collapse to $2^r$ and the coefficient vanishes. In every regime the coefficient equals $\bigl[\binom{r}{L'-1} - \binom{r}{L'}\bigr]_+$: for $r \leq 2L'-1$ the bracketed difference is nonnegative and matches the computations above, while for $r \geq 2L'$ it is negative and the coefficient is $0$. This proves \eqref{eq:Turan_coeff}.

Applying \eqref{eq:Turan_coeff} with $(a', r)$ equal to $(a, b-2)$, $(a-1, b-1)$, $(a-2, b)$ --- corresponding to $(\alpha_t, \alpha_s) = (1,1), (1,2), (2,2)$ --- yields $c_{11}, c_{12}, c_{22}$ respectively.

\medskip
\emph{(ii) Diagonal coefficients.} Substituting \eqref{eq:split} into \eqref{eq:Mtilde_entries_diag},
\[
\widetilde M_{tt} = \underbrace{2\phi_t\, e_N(\boldsymbol{\phi}_{-t})\Phi_N(\boldsymbol{\phi}_{-t})}_{\text{degree 1 in }\phi_t} + \underbrace{2\phi_t^2\, e_N(\boldsymbol{\phi}_{-t})\Phi_{N-1}(\boldsymbol{\phi}_{-t}) + \phi_t^2\, e_N(\boldsymbol{\phi}_{-t})^2}_{\text{degree 2 in }\phi_t}.
\]

\emph{Case $\alpha_t=1$.}
The complementary monomial in $\boldsymbol{\phi}_{-t}$ has $a$ squared variables and $b-1$ first-power variables. Applying \autoref{lem:coefficient_counting} with $K=N$ gives the coefficient
\[
    \binom{b-1}{N-a}
    \mathbf 1_{\{\,2a+b-1-N\le N\,\}}
    =
    \binom{b-1}{L}
    \mathbf 1_{\{\,b-1\le2L\,\}}.
\]
Multiplying by the leading factor $2$ yields~\eqref{eq:d1}.

\emph{Case $\alpha_t=2$.}
The complementary monomial has $a-1$ squared variables and $b$ first-power variables. Applying \autoref{lem:coefficient_counting} to the first summand $2e_N\Phi_{N-1}$ gives
\[
    2\binom{b}{N-(a-1)}
    \mathbf 1_{\{\,2(a-1)+b-N\le N-1\,\}}
    =
    2\binom{b}{L+1}
    \mathbf 1_{\{\,b\le2L+1\,\}}.
\]
For the second summand $e_N^2$, each squared variable must belong to both size-$N$ index sets, while each first-power variable must belong to exactly one of them. Each index set therefore contains $L+1$ of the $b$ first-power variables. This is possible only when $b=2L+2$, in which case there are $\binom{b}{L+1}$ valid ordered pairs. Combining the two contributions yields~\eqref{eq:d2}.

\subsubsection*{Step 3: Verifying $C_\alpha \succeq 0$ across all regimes}

\paragraph{Degenerate cases $L < 0$.}
Each entry of $\widetilde M$ has total degree at most $2N+2$, attained by $\phi_t^2 e_N(\boldsymbol{\phi}_{-t})^2$, so $C_\alpha = 0$ whenever $b + 2a > 2N+2$. If $L \leq -2$, then $b + 2a \geq 2N+4$, so $C_\alpha = 0$. If $L = -1$ and $b \geq 1$, then $b + 2a \geq 2N+3$, so again $C_\alpha = 0$. The only nontrivial case is $L = -1, b = 0$: then $S_2$ has $N+1$ elements, and \eqref{eq:d2}--\eqref{eq:cij} give $d_2 = 1, c_{22} = 1$, so $C_\alpha = I_{N+1} + (J_{N+1} - I_{N+1}) = J_{N+1} \succeq 0$. Henceforth assume $L \geq 0$.

\paragraph{Reduction to three scalar inequalities.}
A block of the form $dI_m+c(J_m-I_m)$ has eigenvalue $d+(m-1)c$ on $\operatorname{span}(\mathbf{1}_m)$ and eigenvalue $d-c$ on its orthogonal complement. In \eqref{eq:Calpha_block}, the coupling term $c_{12}\mathbf{1}_b\mathbf{1}_a^\top$ acts only on $\operatorname{span}(\mathbf{1}_b)\oplus\operatorname{span}(\mathbf{1}_a)$. Therefore, on the within-block orthogonal complements, positive semidefiniteness requires $d_1-c_{11}\geq0$ and $d_2-c_{22}\geq0$. On the two-dimensional coupling subspace, using the normalized basis vectors $\mathbf{1}_b/\sqrt b$ and $\mathbf{1}_a/\sqrt a$, $C_\alpha$ is represented by
\[
\begin{pmatrix}
d_1+(b-1)c_{11} & \sqrt{ab}\,c_{12}\\
\sqrt{ab}\,c_{12} & d_2+(a-1)c_{22}
\end{pmatrix}.
\]
All five coefficients are nonnegative by \eqref{eq:d1}--\eqref{eq:cij}, so the diagonal entries of this matrix are nonnegative. Its positive semidefiniteness is therefore equivalent to a nonnegative determinant. Hence $C_\alpha\succeq0$ if and only if
\begin{equation}\label{eq:PSD_conditions}
d_1-c_{11}\geq0,\qquad
d_2-c_{22}\geq0,\qquad
\bigl[d_1+(b-1)c_{11}\bigr]\bigl[d_2+(a-1)c_{22}\bigr]
\geq ab\,c_{12}^2.
\end{equation}
The first condition is needed only when $b\geq2$, the second only when $a\geq2$, and conditions associated with an absent block are omitted.

\paragraph{First two inequalities of \eqref{eq:PSD_conditions}.}
If $c_{11} > 0$, then $b \leq 2L+1$, so $d_1 = 2\binom{b-1}{L}$. Pascal's identity $\binom{b-1}{L} = \binom{b-2}{L-1} + \binom{b-2}{L}$ gives
\[
d_1 - c_{11} = 2\binom{b-1}{L} - \Bigl[\binom{b-2}{L-1} - \binom{b-2}{L}\Bigr] = \binom{b-2}{L-1} + 3\binom{b-2}{L} \geq 0.
\]
If $c_{11} = 0$, the inequality is trivial. Similarly, if $c_{22} > 0$ and $b \leq 2L+1$, then $d_2 = 2\binom{b}{L+1}$ and
\[
d_2 - c_{22} = 2\binom{b}{L+1} - \Bigl[\binom{b}{L+1} - \binom{b}{L+2}\Bigr] = \binom{b}{L+1} + \binom{b}{L+2} \geq 0.
\]
The boundary subcase $b = 2L+2$ is handled below.

\paragraph{Third inequality of \eqref{eq:PSD_conditions} for $b \leq 2L+1$.}
The cases $a = 0$ (no second block) and $b = 0$ (RHS vanishes) are trivial. If $b = 1$, then $c_{11} = 0$, and $c_{12} = 0$ unless $L = 0$; when $L = 0$, direct substitution gives $d_1 = d_2 = 2$, $c_{12} = c_{22} = 1$, and $d_1(d_2 + (a-1)c_{22}) = 2(a+1) \geq a = ab\,c_{12}^2$.

Henceforth assume $a \geq 1$, $b \geq 2$, $b \leq 2L+1$. Introduce $d \triangleq 2L+2-b \geq 1$ and $x \triangleq \binom{b-1}{L}$. Pascal's identities yield
\[
d_1 = 2x, \quad c_{11} = \tfrac{d-1}{b-1}x, \quad c_{12} = \tfrac{d}{L+1}x, \quad c_{22} = \tfrac{b(d+1)}{(L+1)(L+2)}x, \quad d_2 = \tfrac{2b}{L+1}x,
\]
so that $d_1 + (b-1)c_{11} = (d+1)x$ and $d_2 + (a-1)c_{22} = \tfrac{b}{L+1}x\bigl[2 + \tfrac{(a-1)(d+1)}{L+2}\bigr]$. The third inequality of \eqref{eq:PSD_conditions} becomes
\begin{equation}\label{eq:third_diff}
\tfrac{bx^2}{L+1}\Bigl[(d+1)\bigl(2 + \tfrac{(a-1)(d+1)}{L+2}\bigr) - \tfrac{ad^2}{L+1}\Bigr] \geq 0.
\end{equation}
The bracket is linear in $a$ with slope $\tfrac{(d+1)^2}{L+2} - \tfrac{d^2}{L+1} = \tfrac{-d^2 + 2(L+1)d + (L+1)}{(L+1)(L+2)}$, whose numerator (a downward quadratic in $d$) has positive root $(L+1) + \sqrt{(L+1)(L+2)} > 2L+1 \geq d$, so the slope is nonnegative. At $a = 1$, the bracket equals $2(d+1) - \tfrac{d^2}{L+1} \geq 0$ if and only if $d^2 - 2(L+1)d - 2(L+1) \leq 0$, whose positive root $(L+1) + \sqrt{(L+1)(L+3)}$ likewise exceeds $2L+1$. Hence the bracket is nonnegative at $a = 1$ and weakly increasing thereafter, establishing \eqref{eq:third_diff}.

\paragraph{Boundary case $b = 2L+2$.}
At $b = 2L+2$, the identity $\binom{2L+1}{L} = \binom{2L+1}{L+1}$ forces $c_{12} = 0$, decoupling the two blocks. The first block vanishes ($d_1 = c_{11} = 0$ since $b - 1 > 2L$). The second block has $d_2 = \binom{2L+2}{L+1}$ and $c_{22} = \binom{2L+2}{L+1} - \binom{2L+2}{L+2}$, so $d_2 - c_{22} = \binom{2L+2}{L+2} \geq 0$. All conditions in \eqref{eq:PSD_conditions} hold. For $b > 2L+2$, all coefficients vanish by the indicator constraints, so $C_\alpha = 0$.

\paragraph{Conclusion.}
The four cases above establish $C_\alpha \succeq 0$ for every $\alpha$, hence $\widetilde M(\boldsymbol{\phi}) \succeq 0$ on the positive orthant. By the congruence of Step 1, $M^0 \succeq 0$, completing the proof.

\section{Proofs for \autoref{sec:discussion}}\label{app:disc}

This part collects the proofs of \autoref{strength} and \autoref{prop:salience_cs}.



\subsection*{Proof of \autoref{strength}}

By \autoref{pro:PPP}, $v_{Bt}^*/v_{At}^* = k$ for every $t \in \mathcal{N}$, so the equilibrium battle-winning probabilities \eqref{eq:pStar} are
\[
p_{At}^{*} = \frac{c_t}{c_t + k^{r_t}}, \qquad p_{Bt}^{*} = \frac{k^{r_t}}{c_t + k^{r_t}},
\]
where $c_t = (c_{Bt}/c_{At})^{r_t}$ and $k = W_B/W_A$. Since $p_{At}^*$ depends only on $(c_t, r_t, k)$, team~$A$'s majority-winning probability can be written as
\[
    \textnormal{Prob}_A^{*} = \sum_{\substack{w \subseteq \mathcal{N} \\ |w| \geq N+1}} \prod_{t \in w} p_{At}^{*} \prod_{q \in \mathcal{N} \setminus w} p_{Bq}^{*},
\]
yielding $\partial\, \textnormal{Prob}_A^* / \partial p_{At}^* = \theta^*(t) > 0$, where $\theta^*(t)$ is the pivotality of battle~$t$.

For (i), $\partial p_{At}^*/\partial c_t = k^{r_t}/(c_t + k^{r_t})^2 > 0$, and the chain rule gives $\partial\, \textnormal{Prob}_A^*/\partial c_t > 0$. Since $c_t$ is strictly decreasing in $c_{At}$ and strictly increasing in $c_{Bt}$, the claim follows. For (ii), $\partial p_{At}^*/\partial k = -r_t c_t k^{r_t-1}/(c_t + k^{r_t})^2 < 0$ for every $t$, so $\partial\, \textnormal{Prob}_A^*/\partial k < 0$.

\subsection*{Proof of \autoref{prop:salience_cs}}

The equilibrium battle-winning probabilities are
\[
    p_{At}^{*} \;=\; \frac{c_t}{c_t + k^{r_t}}, \qquad p_{Bt}^{*} \;=\; \frac{k^{r_t}}{c_t + k^{r_t}},
\]
with $c_t = \rho_t^{r_t}$. Since $\theta^*(t)$ depends only on $\{p_{As}^*, p_{Bs}^*\}_{s \neq t}$, it is invariant to perturbations of own-battle primitives that leave the remaining battles unaffected.

\smallskip
\noindent\emph{Part (i).} Fix $r_t$, $k$, and all other primitives. The map $\rho_t \mapsto c_t = \rho_t^{r_t}$ is a strictly increasing bijection on $(0, \infty)$, so it suffices to establish strict single-peakedness in $c_t$ with the maximum at $c_t = k^{r_t}$. Let $x \triangleq c_t$ and $a \triangleq k^{r_t} > 0$. Because $r_t$ and $\theta^*(t)$ are independent of $x$, the dependence of $S_t$ on $x$ is fully captured by
\[
    \mathrm{Clo}_t(x) \;=\; \frac{a\,x}{(x + a)^2}.
\]
Differentiating,
\[
    \frac{d \mathrm{Clo}_t}{d x} \;=\; \frac{a\,(x + a)^2 - 2 a\,x\,(x + a)}{(x + a)^4} \;=\; \frac{a\,(a - x)}{(x + a)^3},
\]
which is strictly positive for $x < a$, zero at $x = a$, and strictly negative for $x > a$. Hence $\mathrm{Clo}_t$, and therefore $S_t$, is strictly single-peaked in $x$ with unique maximum at $x = a$, i.e., $c_t = k^{r_t}$. Pulling back through $c_t = \rho_t^{r_t}$ yields the unique maximum at $\rho_t = k$. Substituting $c_t = k^{r_t}$ into the expression for $p_{At}^*$ gives $p_{At}^* = 1/2$, so the peak coincides with equal equilibrium winning probabilities. Finally, $\mathrm{Clo}_t(x) = a / (x + 2 a + a^2/x) \to 0$ as $x \to 0$ or $x \to \infty$, so $S_t \to 0$ as $\rho_t \to 0$ or $\rho_t \to \infty$.

\smallskip
\noindent\emph{Part (ii).} When $\rho_t=k$, we have $c_t=k^{r_t}$ and hence $\mathrm{Clo}_t=1/4$, independently of $r_t$. Moreover, $\theta^*(t)$ depends only on the other battles and is unaffected by $r_t$. Therefore $S_t=r_t\theta^*(t)/4$ is strictly increasing in $r_t$.

\smallskip
\noindent\emph{Part (iii).} Fix $\{c_s, r_s\}_{s \in \mathcal{N}}$. From $\mathrm{Clo}_t = c_t k^{r_t} / (c_t + k^{r_t})^2$,
\[
    \log \mathrm{Clo}_t \;=\; \log c_t + r_t \log k - 2 \log\!\left(c_t + k^{r_t}\right).
\]
Differentiating with respect to $\log k$, and using $\partial k^{r_t} / \partial \log k = r_t k^{r_t}$,
\[
    \frac{\partial \log \mathrm{Clo}_t}{\partial \log k} \;=\; r_t - 2 \cdot \frac{r_t k^{r_t}}{c_t + k^{r_t}} \;=\; r_t \left(1 - 2 p_{Bt}^*\right) \;=\; r_t\, (p_{At}^* - p_{Bt}^*),
\]
where the last equality uses $p_{At}^* + p_{Bt}^* = 1$. The sign claim is immediate: $\partial \log \mathrm{Clo}_t / \partial \log k > 0$ if and only if $p_{At}^* > p_{Bt}^*$.

\bibliographystyle{aer}
\bibliography{TeamContest_reconstructed_new}

@article{banzhaf1965weighted,
  author  = {Banzhaf, John F.},
  title   = {Weighted Voting Doesn't Work: A Mathematical Analysis},
  journal = {Rutgers Law Review},
  year    = {1965},
  volume  = {19},
  pages   = {317--343}
}

@article{Barbieri-2024-Winners,
  author  = {Barbieri, Stefano and Serena, Marco},
  title   = {Winner's Effort in Multi-Battle Team Contests},
  journal = {Games and Economic Behavior},
  year    = {2024},
  volume  = {145},
  pages   = {526--556}
}

@article{brams1974three,
  author  = {Brams, Steven J. and Davis, Morton D.},
  title   = {The 3/2's Rule in Presidential Campaigning},
  journal = {American Political Science Review},
  year    = {1974},
  volume  = {68},
  number  = {1},
  pages   = {113--134}
}

@article{Clark-2020-Creating,
  author  = {Clark, Derek J. and Nilssen, Tore},
  title   = {Creating Incentives at the Margin: The Case of Dynamic Multi-Battle Contests},
  journal = {Journal of Public Economic Theory},
  year    = {2020},
  volume  = {22},
  pages   = {1529--1551}
}

@article{dixit1996determinants,
  author  = {Dixit, Avinash and Londregan, John},
  title   = {The Determinants of Success of Special Interests in Redistributive Politics},
  journal = {Journal of Politics},
  year    = {1996},
  volume  = {58},
  number  = {4},
  pages   = {1132--1155}
}

@article{dynes2015partisanship,
  author  = {Dynes, Adam M. and Huber, Gregory A.},
  title   = {Partisanship and the Allocation of Federal Spending: Do Same-Party Legislators or Voters Benefit from Shared Party Affiliation with the President and House Majority?},
  journal = {American Political Science Review},
  year    = {2015},
  volume  = {109},
  number  = {1},
  pages   = {172--186}
}

@article{fang2020turning,
  author  = {Fang, Dawei and Noe, Thomas and Strack, Philipp},
  title   = {Turning Up the Heat: The Discouraging Effect of Competition in Contests},
  journal = {Journal of Political Economy},
  year    = {2020},
  volume  = {128},
  number  = {9},
  pages   = {3401--3430}
}

@article{Feng-2018-How,
  author  = {Feng, Xin and Lu, Jingfeng},
  title   = {How to Split the Pie: Optimal Rewards in Dynamic Multi-Battle Competitions},
  journal = {Journal of Public Economics},
  year    = {2018},
  volume  = {160},
  pages   = {82--95}
}

@article{Feng-2024-Optimal,
  author  = {Feng, Xin and Jiao, Qian and Kuang, Zhonghong and Lu, Jingfeng},
  title   = {Optimal Prize Design in Team Contests with Pairwise Battles},
  journal = {Journal of Economic Theory},
  year    = {2024},
  volume  = {215},
  pages   = {105765}
}

@article{Fu-2015-Team,
  author  = {Fu, Qiang and Lu, Jingfeng and Pan, Yue},
  title   = {Team Contests with Multiple Pairwise Battles},
  journal = {American Economic Review},
  year    = {2015},
  volume  = {105},
  number  = {7},
  pages   = {2120--2140}
}

@article{Gelder-2014-From,
  author  = {Gelder, Alan},
  title   = {From Custer to Thermopylae: Last Stand Behavior in Multi-Stage Contests},
  journal = {Games and Economic Behavior},
  year    = {2014},
  volume  = {87},
  pages   = {442--466}
}

@article{Hafner-2017-Tug,
  author  = {H{\"a}fner, Samuel},
  title   = {A Tug-of-War Team Contest},
  journal = {Games and Economic Behavior},
  year    = {2017},
  volume  = {104},
  pages   = {372--391}
}

@article{Hafner-2020-Eternal,
  author  = {H{\"a}fner, Samuel},
  title   = {Eternal Peace in the Tug-of-War?},
  journal = {Economic Theory},
  year    = {2022},
  volume  = {74},
  number  = {4},
  pages   = {1057--1101}
}

@article{halac2021rank,
  author  = {Halac, Marina and Lipnowski, Elliot and Rappoport, Daniel},
  title   = {Rank Uncertainty in Organizations},
  journal = {American Economic Review},
  year    = {2021},
  volume  = {111},
  number  = {3},
  pages   = {757--786}
}

@article{Harris-1987-Racing,
  author  = {Harris, Christopher and Vickers, John},
  title   = {Racing with Uncertainty},
  journal = {Review of Economic Studies},
  year    = {1987},
  volume  = {54},
  pages   = {1--21}
}

@article{holmstrom1982moral,
  author  = {Holmstr{\"o}m, Bengt},
  title   = {Moral Hazard in Teams},
  journal = {Bell Journal of Economics},
  year    = {1982},
  volume  = {13},
  number  = {2},
  pages   = {324--340}
}

@article{Klumpp-2006-Primaries,
  author  = {Klumpp, Tilman and Polborn, Mattias K.},
  title   = {Primaries and the New Hampshire Effect},
  journal = {Journal of Public Economics},
  year    = {2006},
  volume  = {90},
  number  = {6--7},
  pages   = {1073--1114}
}

@article{Kobayashi-2024-Effort,
  author  = {Kobayashi, Katsuya},
  title   = {Effort Complementarity and Role Assignments in Group Contests},
  journal = {Journal of Economics \& Management Strategy},
  year    = {2024},
  volume  = {33},
  number  = {3},
  pages   = {483--508}
}

@article{Konishi-2022-Equilibrium,
  author  = {Konishi, Hideo and Pan, Chen and Simeonov, Dimitar},
  title   = {Equilibrium Player Choices in Team Contests with Multiple Pairwise Battles},
  journal = {Games and Economic Behavior},
  year    = {2022},
  volume  = {132},
  pages   = {274--287}
}

@article{Konishi-2024-Allocation,
  author  = {Konishi, Hideo and Sahuguet, Nicolas and Crutzen, Beno{\^i}t S. Y.},
  title   = {Allocation Rules of Indivisible Prizes in Team Contests},
  journal = {Economic Theory},
  year    = {2024},
  volume  = {78},
  pages   = {69--100}
}

@article{Konrad-2009-Multi,
  author  = {Konrad, Kai A. and Kovenock, Dan},
  title   = {Multi-Battle Contests},
  journal = {Games and Economic Behavior},
  year    = {2009},
  volume  = {66},
  number  = {1},
  pages   = {256--274}
}

@article{Kovenock-2021-Generalizations,
  author  = {Kovenock, Dan and Roberson, Brian},
  title   = {Generalizations of the General Lotto and Colonel Blotto Games},
  journal = {Economic Theory},
  year    = {2021},
  volume  = {71},
  pages   = {997--1032}
}

@article{kujala2020donors,
  author  = {Kujala, Jordan},
  title   = {Donors, Primary Elections, and Polarization in the United States},
  journal = {American Journal of Political Science},
  year    = {2020},
  volume  = {64},
  number  = {3},
  pages   = {587--602}
}

@article{lazear1981rank,
  author  = {Lazear, Edward P. and Rosen, Sherwin},
  title   = {Rank-Order Tournaments as Optimum Labor Contracts},
  journal = {Journal of Political Economy},
  year    = {1981},
  volume  = {89},
  number  = {5},
  pages   = {841--864}
}

@article{lemus2025contingent,
  author  = {Lemus, Jorge and Marshall, Guillermo},
  title   = {Contingent Prizes in Dynamic Contests},
  journal = {RAND Journal of Economics},
  year    = {2025},
  volume  = {56},
  number  = {4},
  pages   = {476--493}
}

@article{moldovanu2001optimal,
  author  = {Moldovanu, Benny and Sela, Aner},
  title   = {The Optimal Allocation of Prizes in Contests},
  journal = {American Economic Review},
  year    = {2001},
  volume  = {91},
  number  = {3},
  pages   = {542--558}
}

@article{Nitzan-2011-Prize,
  author  = {Nitzan, Shmuel and Ueda, Kaoru},
  title   = {Prize Sharing in Collective Contests},
  journal = {European Economic Review},
  year    = {2011},
  volume  = {55},
  pages   = {678--687}
}

@article{Nitzan-2014-Intra,
  author  = {Nitzan, Shmuel and Ueda, Kaoru},
  title   = {Intra-Group Heterogeneity in Collective Contests},
  journal = {Social Choice and Welfare},
  year    = {2014},
  volume  = {43},
  pages   = {219--238}
}

@article{olszewski2016large,
  author  = {Olszewski, Wojciech and Siegel, Ron},
  title   = {Large Contests},
  journal = {Econometrica},
  year    = {2016},
  volume  = {84},
  number  = {2},
  pages   = {835--854}
}

@article{penrose1946elementary,
  author  = {Penrose, L. S.},
  title   = {The Elementary Statistics of Majority Voting},
  journal = {Journal of the Royal Statistical Society},
  year    = {1946},
  volume  = {109},
  number  = {1},
  pages   = {53--57}
}

@article{rayo2007relational,
  author  = {Rayo, Luis},
  title   = {Relational Incentives and Moral Hazard in Teams},
  journal = {Review of Economic Studies},
  year    = {2007},
  volume  = {74},
  number  = {3},
  pages   = {937--963}
}

@article{Roberson-2006-Colonel,
  author  = {Roberson, Brian},
  title   = {The Colonel Blotto Game},
  journal = {Economic Theory},
  year    = {2006},
  volume  = {29},
  number  = {1},
  pages   = {1--24}
}

@article{shapley1954method,
  author  = {Shapley, L. S. and Shubik, Martin},
  title   = {A Method for Evaluating the Distribution of Power in a Committee System},
  journal = {American Political Science Review},
  year    = {1954},
  volume  = {48},
  number  = {3},
  pages   = {787--792}
}

@article{stromberg2008electoral,
  author  = {Str{\"o}mberg, David},
  title   = {How the Electoral College Influences Campaigns and Policy: The Probability of Being Florida},
  journal = {American Economic Review},
  year    = {2008},
  volume  = {98},
  number  = {3},
  pages   = {769--807}
}

@article{thomsen2023competition,
  author  = {Thomsen, Danielle M.},
  title   = {Competition in Congressional Elections: Money versus Votes},
  journal = {American Political Science Review},
  year    = {2023},
  volume  = {117},
  number  = {2},
  pages   = {675--691}
}

@article{winter2004incentives,
  author  = {Winter, Eyal},
  title   = {Incentives and Discrimination},
  journal = {American Economic Review},
  year    = {2004},
  volume  = {94},
  number  = {3},
  pages   = {764--773}
}

\end{document}